\documentclass[a4paper]{article}
\usepackage{comment}
\usepackage[english]{babel}
\usepackage[utf8x]{inputenc}
\usepackage[T1]{fontenc}
\usepackage{amssymb}
\usepackage{amsthm}
\usepackage{multicol}
\usepackage[a4paper,top=2cm,bottom=1.5cm,left=2cm,right=2cm,marginparwidth=2cm]{geometry}

\usepackage{enumerate}

\def\R{{\cal R}}

\newcommand{\diag}{\operatorname{diag}}

\newcommand{\T}{\operatorname{T}}
\newcommand{\tr}{\operatorname{tr}}
\newcommand{\re}{\operatorname{R}}
\newcommand{\Su}{\operatorname{S}}
\newcommand{\In}{\operatorname{I}}
\newcommand{\xx}{\operatorname{X}}
\newcommand{\lag}{\operatorname{lag}}
\newcommand{\eul}{\operatorname{eul}}

\newcounter{acounter}

\usepackage[export]{adjustbox}
\usepackage{subfig}
\usepackage{amsmath}
\usepackage{graphicx}
\usepackage[colorinlistoftodos]{todonotes}
\usepackage[colorlinks=true, allcolors=blue]{hyperref}
\usepackage{float}
\usepackage{tikz}

\usepackage{amsmath,amsfonts,amssymb,amsthm}
\usepackage{mathtools}
\usepackage{commath}
\let\norm\undefined % <-- "Undefine" \norm
\DeclarePairedDelimiter\norm{\lVert}{\rVert}

\providecommand{\keywords}[1]
{
  \small	
  \textbf{\textit{Keywords---}} #1
}

\title{Relating Eulerian and Lagrangian spatial models for vector-host disease dynamics through a fundamental matrix}
\author{Esteban Vargas Bernal$^{1}$\footnote{Corresponding author. E-mail: vargasbernal.1@osu.edu.}, Omar Saucedo$^{2}$, Joseph Hua Tien$^{1}$  \\
        \small $^{1}$Department of Mathematics, The Ohio State University. \\
        \small $^{2}$Department of Mathematics, Virginia Tech. \\
}

\newtheorem{definition}{Definition}
\newtheorem{proposition}{Proposition}
\newtheorem*{proposition*}{Proposition}

\begin{document}
\maketitle

\begin{abstract}
\noindent  We explore the relationship between Eulerian and Lagrangian approaches for modeling movement in vector-borne diseases for discrete space. In the Eulerian approach we account for the movement of hosts explicitly through movement rates captured by a graph Laplacian matrix $L$. In the Lagrangian approach we only account for the proportion of time that individuals spend in foreign patches through a mixing matrix $P$. We establish a relationship between an Eulerian model and a Lagrangian model for the hosts in terms of the matrices $L$ and $P$. 
We say that the two modeling frameworks are consistent if for a given matrix $P$, the matrix $L$ can be chosen so that the residence times of the matrix $P$ and the matrix $L$ match.
We find a sufficient condition for consistency, and examine disease quantities such as the final outbreak size and basic reproduction number in both the consistent and inconsistent cases.
In the special case of a two-patch model, we observe how similar values for the basic reproduction number and final outbreak size can occur even in the inconsistent case. However, there are scenarios where the final sizes in both approaches can significantly differ by means of the relationship we propose. 

\end{abstract}

\begin{center}
\end{center}

%%%%%%%%%%%%%%%%%%%%%%%%%%%%%%%%%%%%%%%%%%%
\section{Introduction}
%%%%%%%%%%%%%%%%%%%%%%%%%%%%%%%%%%%%%%%%%%%

Understanding how connectivity between different spatial locations affects vector-borne disease dynamics is a fundamental issue in disease ecology and public health. In particular, wide variation in disease transmission between locales is commonplace, reflecting heterogeneity in breeding site availability, geography, climate, availability of bed nets and window screens, demography, and many other factors.  How this spatial heterogeneity interacts with connectivity through host and vector movement to inform disease dynamics is not obvious. For example, empirical studies have shown how a vector-borne disease may persist in cities where mosquito abundance is low or zero (for example for malaria in \cite{ochoa2006epidemiology}). Some authors have shown that this persistence can be explained by host movement.  Stoddard et al. \cite{stoddard2009role} used a conceptual model to show that even when the vector density is low, the risk of acquiring the disease may be high due to movement. Similarly, Cosner et al. \cite{cosner2009effects} constructed a spatial vector-borne disease model to study how human movement can affect transmission, and discovered that human movement between heterogeneous locations was sufficient to sustain disease persistence. In other situations, movement can lead to disease extinction, even in areas with high local transmission \cite{Tatem2010PopMove}. Movement patterns also affect the spatial spread of vector-borne diseases such as Lyme disease \cite{gaff2007modeling}, West Nile Virus \cite{liu2006modeling}, Dengue \cite{Guido2019Dengue}, Zika \cite{OReilly2018Zika, Zhang2017Zika}, and Malaria \cite{Tatem2010PopMove,Wesolowski2012MalariaMove}.  Both connectivity and local conditions for disease transmission 
are important considerations when designing disease surveillance and control efforts \cite{woolhouse1997heterogeneities}.

Many different approaches have been taken for modeling spatial vector-borne disease dynamics, including 
PDEs \cite{Lewis2006Traveling}, ODEs \cite{Acevedo2015Spatial,Pindolia2012HumanMove, ruktanonchai2016identifying}, stochastic models \cite{Jovanovic2013StochasticVB, Wanduku2012}, and agent-based models \cite{Bomblies2014ABMMalaria, Jindal2017ABMVB}. A widely-used building block for modeling mosquito-borne disease is the Ross-MacDonald model \cite{reiner2013systematic,ross1916application}.  See \cite{smith2004statics} for derivation of the Ross-Macdonald system that is considered here. This paper concerns extensions of the Ross-MacDonald framework to include multiple discrete spatial locations. These locations might correspond to villages, cities, health districts, or the like, with linkages between them through movement of host and vector.  Note that the connectivity  patterns for host and vector may be different, for example reflecting different movement scales of each. Regarding mosquito-borne diseases, Service \cite{service1997mosquito} provides a review of the types of mosquitoes movement (long/short dispersal), which may vary significantly among different species. 
Empirical studies using different capture methods (bed net catches, exit trap catches, oviposition traps) allow us to have an idea of the spatial scale of mosquito movement \cite{honorio2003dispersal,thomson1995movement}.
Given the importance that host and vector movement 
may have for the spread of a vector-borne disease, in this paper we consider both movements.

We now encounter a dichotomy in the modeling approaches: whether to treat individuals (host or vector) as residents of a particular patch and commuting to the others, versus migration between patches without a fixed sense of home.  Following the terminology of \cite{cosner2015models, cosner2009effects}, we will refer to the former approach as {\it Lagrangian}, and the latter as {\it Eulerian}.  This terminology stems from similarities to the Lagrangian and Eulerian modeling approaches for the description of fluid motion in fluid mechanics \cite{krause2005fluid}. 
For more on Eulerian and Lagrangian approaches for modeling movement, see \cite{grunbaum1994,gueron1996}.
Each of these approaches has its strengths. Lagrangian models are often a natural choice on small spatial scales and in settings where individuals have a sense of home.  Eulerian models may be suitable on large spatial scales, and in migratory settings where the origin of the individuals is less important than their current location.  This may include situations where we are interested in introduction, reintroduction, or global spread of disease, as discussed in \cite{stoddard2009role}.  

Eulerian and Lagrangian approaches are widely-used to model vector-borne disease dynamics on discrete space. See \cite{dye1986population, hasibeder1988population,rodriguez2001models,ruktanonchai2016identifying} for some studies that have used the Lagrangian approach. For example, in \cite{dye1986population,hasibeder1988population}, each study used a Lagrangian framework to examine how heterogeneity in the distribution of \emph{Anopheles} mosquitoes could lead to a larger  reproduction number compared to the setting when the mixing is homogeneous.
See also 
\cite{allen2007asymptotic,gaff2007modeling,hsieh2007impact,liu2006modeling} for applications of the Eulerian approach. For example, in \cite{liu2006modeling} an Eulerian model was used to study how the long range dispersal of birds may explain discontinuities in the spread of West Nile Virus. 
Combinations of the Eulerian and Lagrangian approaches are taken in \cite{arino2005multi,arino2003multi, arino2006DiseaseMeta,iggidr2017vector}. 

What are the functional implications of using one approach versus the other?  Do the Eulerian and Lagrangian approaches yield similar results regarding important disease quantities such as the basic reproduction number and outbreak size?  Eulerian models tend to be analytically tractable, for example allowing establishment of asymptotic disease profiles \cite{allen2007asymptotic}, global stability \cite{shuai2013}, and application of techniques from spectral graph theory to estimate the basic reproduction number \cite{tien2015disease}. 
We would like to know whether one can safely use an Eulerian approach to model vector-borne disease in a setting where Lagrangian data are available (e.g., perhaps in terms of proportions of time spent in different locations) or where spatial scales are small and commuting (vs. migration) is typical.  Comparing the Eulerian and Lagrangian approaches to modeling movement, and the functional implications of using one approach versus the other in vector-borne disease models, is the focus of this paper.

Consider vector-host disease dynamics on $n$ spatial locations.
We define two families of models for this setting, where both use an Eulerian framework for modeling vector movement. In the first model we consider a Lagrangian approach for host movement, and we refer to this as the \textit{Lagrangian model} \cite{cosner2009effects,martcheva2015introduction}. In the second model we consider an Eulerian framework for host movement, and we refer to this model as the \textit{Eulerian model} \cite{cosner2009effects,martcheva2015introduction}. The Lagrangian character of the first model is captured by a mixing matrix $P = (p_{ij})_{i,j\leq n}$, where $p_{ij}$ is the fraction of time that a resident of patch $j$ spends in patch $i$. The Eulerian character of the second model is captured by an adjacency matrix $A = (m_{ij})_{i,j\leq n}$, with $m_{ij}$ being the per capita movement rate from patch $j$ to patch $i$. 
We will work extensively with the unnormalized  (`combinatorial') graph Laplacian $L=W-A$, where $W$ diagonal with $ W_{ii}= \sum_{k=1}^n A_{ki}$.
The graph Laplacian $L$ is a basic object in graph theory that conveys a great deal of structural information about the network associated with $A$, including the number of connected components, spanning trees, community structure, and more \cite{chung1997spectral, ng2002spectral,von2007tutorial}.

The main goal of this paper is studying the relationship between the Lagrangian and Eulerian models. 
Specifically, we obtain a relationship between the matrices $L$ and $P$ through a fundamental matrix that captures the expected time that an individual from one location spends in another. We give criteria for when the Eulerian and Lagrangian frameworks are consistent, meaning that the two frameworks can exactly match in terms of this fundamental matrix, and consider the functional implications for disease dynamics in both the consistent and inconsistent settings.  These results can serve as a guide for when one framework can be substituted for the other.

The following is the distribution of the content of this paper. In Section \ref{smodel} we define the Lagrangian and Eulerian systems that we study, and we give the disease-free equilibria and the basic reproduction numbers $\mathcal{R}_0^{\lag}$, $\mathcal{R}_0^{\eul}$ for both systems. In Section \ref{SectionPL} we relate the Lagrangian and Eulerian systems through a minimization problem involving the matrices $P$ and $L$.  We say that the two systems are {\it consistent} if this minimization problem can be solved exactly.
Following some preliminaries in Section \ref{sec:prelims}, we formulate this minimization problem in Section \ref{sec:formulation}. In Section \ref{sec:consistency} we look at the consistent scenario. We first provide a sufficient condition for the consistency of the proposed relationship in Section \ref{Ssufficientcondition}, and then, we give an example of a consistent relationship in Section \ref{Sconsistentexample}. In Section \ref{sectioninconsistent}
we give examples where the relationship is inconsistent. In Section \ref{section2by2} we study the relationship between the Lagrangian and Eulerian approaches for a simple network consisting of two patches. In Section \ref{sec:finalsize2patches} we compare the final outbreak sizes
and basic reproduction numbers of both systems under an inconsistent example. In Section \ref{sec:comparisonhomo} we compare the basic reproduction numbers of consistent examples when we vary the entries of the mixing matrix.  In Section \ref{SdataP} we explore some examples of matrices $P$ from empirical and hypothetical data. In Section \ref{sec:conclusions} we present the main conclusions of the paper. Finally, we give the details of some results of the previous sections in an appendix in Section \ref{Sappendix}.

%%%%%%%%%%%%%%%%%%%%%%%%%%%%%%%%%%%%%%%%%%%
\section{Modeling frameworks} \label{smodel}
%%%%%%%%%%%%%%%%%%%%%%%%%%%%%%%%%%%%%%%%%%%

The basic building block for the modeling frameworks considered in this paper is the Ross-MacDonald vector-host model, as considered by \cite{smith2004statics}.  For a single spatial location, the model equations are:
\begin{equation}
\label{eqn:rossmac}
\begin{array}{ccl}
\dot{S} &=& \Lambda - \beta \frac{S}{N} I_v - \mu S\\
\dot{I} &=& \beta \frac{S}{N}I_v - (\gamma + \mu) I\\
\dot{R} &=& \gamma I - \mu R \\
\dot{S}_v &=& \Lambda_v - \beta_v \frac{I}{N} S_v - \mu_v S_v \\
\dot{I}_v &=&  \beta_v \frac{I}{N} S_v - \mu_v I_v \,.
\end{array}
\end{equation}

System \eqref{eqn:rossmac} takes a Susceptible-Infectious-Recovered (SIR) framework for host and Susceptible-Infectious (SI) framework for vector.  $S$ represents the number of susceptible hosts, $I$ represents the number of infectious hosts, and $R$ represents the number of recovered hosts.  Similarly, $S_{v}$ denotes the number of susceptible vectors and  $I_{v}$ denotes the number of infected vectors.  The total host population is denoted by $N$. $\Lambda$ and $\Lambda_v$ are the constant recruitment rates for host and vector, respectively. The transmission rate from vector to host is $\beta_{i}$ and the transmission rate from host to vector is $\beta_{v,i}$.  The parameters $\mu$ and $\mu_{v}$ correspond to the mortality rates for host and vectors.  Many modifications to this framework are possible, including incubation periods, seasonal forcing, and much more (see \cite{childs2009interaction} for an example of seasonal forcing and \cite{smith2012ross} for a more general review).  We consider here the very simple system \eqref{eqn:rossmac} in order to focus on the impact of connectivity between different spatial locations.

Consider $n$ distinct spatial locations, each with local Ross-MacDonald dynamics as in \eqref{eqn:rossmac} with patch-specific parameters.  We will consider two modeling frameworks that differ in how the spatial locations are coupled through host movement: one using a Lagrangian approach, and the other an Eulerian approach \cite{cosner2009effects,martcheva2015introduction}. 

In the Lagrangian approach, coupling is a {\it mixing matrix} $P = (p_{ij})_{i,j\leq n}$, where $p_{ij}$ is the proportion of time a resident of patch $j$ spends in patch $i$. By contrast, 
in the Eulerian approach, coupling is via the adjacency matrices $M^{\xx} = (m_{ij}^{\xx})_{i,j\leq n}$ of weighted, directed graphs, where $m_{ij}^{\xx}$ is the per capita migration rate of hosts in state $\xx \in \{\Su,\In,\re\}$.

In both modeling frameworks, vector movement is modeled using an Eulerian framework, with $M^v$ the weighted adjacency matrix for vector movement.  Thus the two frameworks considered are Lagrangian (host) / Eulerian (vector), and Eulerian (host) / Eulerian (vector).  For brevity we will refer to these frameworks as simply Lagrangian in the former, and Eulerian in the latter.

In the ensuing analysis we will make extensive use of the (unnormalized) graph Laplacian $L$ \cite{aggarwal2014data}.  Let $A$ be the adjacency matrix for a (weighted, directed) graph $G$, with $A_{ij}$ the weight of the edge from $j$ to $i$.  Let $W$ be the diagonal out-degree matrix with $W_{ii} = \sum_{i=1}^n A_{ij}$.  Then the graph Laplacian is defined as
\begin{equation}
\label{eqn:laplacian}
L = W-A \, .
\end{equation}

The graph Laplacian is a fundamental quantity in graph theory that conveys structural information about $G$, including the number of connected components, spanning trees, community structure, and more  \cite{aggarwal2014data, chung1997spectral, ng2002spectral}.  In the context of infectious disease dynamics, the Laplacian arises naturally in the calculation of $\R_0$ for migration models \cite{tien2015disease}.

All model parameters throughout are assumed to be non-negative.  The adjacency matrices $M^{\xx}$, for $\xx \in {\{\Su,\In,\re\}}$, and $M^v$ are assumed to have zero diagonal (the corresponding graphs for Eulerian movement have no self-edges).  By definition, the columns of the mixing matrix $P$ sum to one.

%%%%%%%%%%%%%%%%%%%%%%%%%%%%%%%%%%%%%%%%%%%
\subsection{Lagrangian for host, Eulerian for vector}\label{sect:lag}
%%%%%%%%%%%%%%%%%%%%%%%%%%%%%%%%%%%%%%%%%%%

The Lagrangian approach can be viewed as a multigroup model with groups corresponding to hosts that are residents of the different spatial locations. The number of susceptible, infectious and recovered individuals that are residents of patch $i$ are denoted by $S_i$, $I_i$ and $R_i$ respectively, and the number of susceptible and infectious vectors in patch $i$ are denoted by $S_{v,i}$ and $I_{v,i}$ respectively. The proportion $p_{ij}$ is the ratio between the time spent by a resident of $j$ in patch $i$ to the whole time spent by a resident of $j$ in all the visited patches.
All of these proportions are collected by the mixing matrix $P=(p_{ij})_{i,j\leq n}$. The movement rate of a vector from patch $j$ to patch $i$ is given by $m_{ij}^v$, and all these rates are collected in the adjacency matrix $M^{v} = (m_{ij}^v)_{i,j\leq n}$. The transmission rates $\beta_i$ (from vectors to hosts) and $\beta_{v,i}$ (from hosts to vectors) are intrinsic to the corresponding patch and determine the transmission rates in (\ref{emodel1}) by averaging according to the mixing rates. Namely, the transmission rate for host residents of patch $i$ is $\sum_{j=1}^n \beta_j p_{ji}\left(S_i/N_i\right)I_{v,j}$ and the transmission rate for vectors in patch $i$ is $\beta_{v,i} \left[\left( \sum_{j=1}^n p_{ij}I_j\right) \middle/\left(\sum_{j=1}^np_{ij}N_j\right)\right] S_{v,i}$. The transmission rate $\sum_{j=1}^n \beta_j p_{ji}\left(S_i/N_i\right)I_{v,j}$ has been considered in Lagrangian models such as in \cite{cosner2009effects,dye1986population,rodriguez2001models,ruktanonchai2016identifying}. There exist other models that incorporate more complex transmission rates for the host. For example, \cite{bichara2016vector} considered $\sum_{k=1}^np_{jk}N_k$ instead of $N_i$ in the denominator of the transmission rate for host individuals under other assumptions. However, the transmission rate considered here is justified by its inclusion in other studies and its analytical tractability.  
The transmission rate $\beta_{v,i} \left[\left( \sum_{j=1}^n p_{ij}I_j\right) \middle/\left(\sum_{j=1}^np_{ij}N_j\right)\right] S_{v,i}$ for vectors has been adopted in studies such as \cite{ruktanonchai2016identifying}. Moreover, the whole Lagrangian system that we consider was also studied in \cite{ruktanonchai2016identifying}. 

The recruitment, mortality and recovery rates for host residents of patch $i$ are $\Lambda_i^{\lag}, \mu_i^{\lag},$ and $\gamma_i^{\lag}$, respectively. The recruitment and mortality rates for vectors in patch $i$ are $\Lambda_{v,i}$ and $\mu_{v,i},$ respectively. 
Let $\delta_i^{\lag}= \mu_i^{\lag}+\gamma_i^{\lag}$ and $\delta_{v,i} = \mu_{v,i}$ denote the host and vector removal rates, respectively. Equations for the Lagrangian framework are given in system \eqref{emodel1}:
\begin{equation}\label{emodel1}
\begin{array}{ccl}
\dot{S_i}& = & \Lambda_i^{\lag} - \sum_{j=1}^n \beta_j p_{ji}\frac{S_i}{N_i}I_{v,j} - \mu_i^{\lag} S_i \\
\dot{I_i}& = &\sum_{j=1}^n \beta_j p_{ji}\frac{S_i}{N_i}I_{v,j} - \left(\gamma_i^{\lag} + \mu_i^{\lag}\right)I_i \\
\dot{R_i}& = &\gamma_i^{\lag} I_i - \mu_i^{\lag} R_i \\
\dot{S}_{v,i}& = &\Lambda_{v,i} -\beta_{v,i} \frac{\sum_{j=1}^n p_{ij}I_j}{\sum_{j=1}^np_{ij}N_j} S_{v,i} + \sum_{j=1}^n m_{ij}^v S_{v,j} - \sum_{j=1}^n m_{ji}^v S_{v,i} - \mu_{v,i} S_{v,i} \\
\dot{I}_{v,i}& = &\beta_{v,i} \frac{\sum_{j=1}^n p_{ij}I_j}{\sum_{j=1}^np_{ij}N_j} S_{v,i} + \sum_{j=1}^n m_{ij}^v I_{v,j} - \sum_{j=1}^n m_{ji}^v I_{v,i} - \mu_{v,i} I_{v,i}\, ,
\end{array}
\end{equation}
for $i=1,\dots,n$.

We assume the following throughout:
\begin{list}{{\bf A\arabic{acounter}:~}}{\usecounter{acounter}}
    \item The adjacency matrix for vector movement has zeros on the diagonal (i.e. $m^v_{ii} = 0$ for all $i$).
        \label{a:M}
    \item The parameters $\beta_i, \beta_{v,i}$ are non-negative.
        \label{a:pos}
    \item The parameters $\Lambda_i^{\lag},\Lambda_{v,i},\gamma_i^{\lag}$, $\mu_i^{\lag}$ and $\mu_{v,i}$ are all positive.
        \label{a:nonneg}
    \item The mixing matrix $P$ is non-negative, with $\sum_{i=1}^n p_{ij} = 1$ for all $j$.
        \label{a:P}    
\end{list}

Table \ref{param1}  contains the parameters used for system (\ref{emodel1}).

\begin{table}[H]
\caption{Parameters for systems  (\ref{emodel1}) and (\ref{emodel2}). 
}
\label{param1}
\begin{tabular}{lll}
\hline\noalign{\smallskip}
Parameter & Meaning & Units  \\
\noalign{\smallskip}\hline\noalign{\smallskip}
$\Lambda_i^{\lag}, \Lambda_i^{\eul}$ & Recruitment rate of susceptible host in patch $i$ & Hosts $\times$ Days$^{-1}$ \\
$\Lambda_{v,i}$ & Recruitment rate of susceptible vectors in patch $i$ & Vectors $\times$ Days$^{-1}$ \\
$\gamma_i^{\eul}$ & Per capita recovery rate of hosts in patch $i$. & Days$^{-1}$  \\
$\mu_i^{\eul}$ & Per capita mortality rate of hosts in patch $i$ & Days$^{-1}$ \\
$\delta_i^{\eul}$ & Per capita removal rate of infectious hosts in patch $i$ & Days$^{-1}$  \\
$\gamma_i^{\lag}$ & Per capita recovery rate of hosts from patch $i$ & Days$^{-1}$  \\
$\mu_i^{\lag}$ & Per capita mortality rate of hosts from patch $i$ & Days$^{-1}$ \\
$\delta_i^{\lag}$ & Per capita removal rate of infectious hosts from patch $i$ & Days$^{-1}$   \\
$p_{ji}$ & Proportion of time that a host from patch $i$ spends in patch $j$ & Dimensionless  \\
$\beta_i $  & Transmission rate to hosts per vector in patch $i$ & Hosts  $\times$ Days$^{-1}\times$ Vectors$^{-1}$  \\
$\beta_{v,i} $ & Per capita transmission rate to vectors in patch $i$ & Days$^{-1}$  \\
$m_{ji}^{\xx}$ & Per capita movement rate of hosts in state $\xx$  & Days$^{-1}$   \\
& from patch $i$ to $j$, for $\xx\in\{\Su,\In,\re\}$ & \\
$m_{ji}^v$ & Per capita movement rate of vectors from patch $i$ to $j$&  Days$^{-1}$   \\

\noalign{\smallskip}\hline
\end{tabular}
\end{table}

Let $G_v := L_v + D_{\delta_v}$, where $L_v$ is the graph Laplacian corresponding to the  adjacency matrix $M^v$ that captures the vector movement, $D_{\delta_v}:= \diag\left\{\delta_{v,i}\right\}$ and $D_\mu^{\lag} := \diag\left \{\mu_i^{\lag}\right\}$. 
Note that $G_v$ has the Z-sign pattern \cite{berman1994nonnegative}, and under assumption A\ref{a:pos} has positive column sums. Thus $G_v$ is a non-singular $M$-matrix \cite{berman1994nonnegative}. 
The disease-free equilibrium (DFE) of the susceptible compartments of  model (\ref{emodel1}) is then

\begin{equation}\label{eqn:DFElagrangian}
\begin{array}{rcl}
\left(S^{\lag}\right)^* = \left(N^{\lag}\right)^* &=&
\left(D_\mu^{\lag}\right)^{-1}\Lambda^{\lag} \, ,\\ \left(S_v^{\lag}\right)^*=\left(N_v^{\lag}\right)^* &=& G_v^{-1}\Lambda_v \, ,
\end{array}
\end{equation}
where $\left(S^{\lag}\right)^*, \left(N^{\lag}\right)^*$, $\left(S_v^{\lag}\right)^*, \left(N_v^{\lag}\right)^*$, $\Lambda_v $ and $\Lambda^{\lag}$ are column vectors with components
$\left(S_{i}^{\lag}\right)^*, \left(N_{i}^{\lag}\right)^*$, $\left(S_{v,i}^{\lag}\right)^*, \left(N_{v,i}^{\lag}\right)^*$, $\Lambda_{v,i}$  and  $\Lambda_{i}^{\lag},$ respectively. (The superscript $*$ indicates evaluation at the DFE.) 

Consider the basic reproduction number $\R_0^{\lag}$ for system \eqref{emodel1}, computed using the next generation matrix approach \cite{van2002reproduction}.  
Then $\left(\R_0^{\lag}\right)^2 = \rho\left(\left(F^{\lag}\right)\left(V^{\lag}\right)^{-1}\right)$, where $\rho$ denotes the spectral radius, and
$F^{\lag}$ and $V^{\lag}$ denote the fecundity and transfer matrices for system \eqref{emodel1}.  Let  
$D_{\beta}:= \diag\{\beta_i\}$, $\hat{N}^*:= P(N^{\lag})^*$, $D_{\beta_v}^{\lag} := \diag \left\{ \beta_{v,i} \left(N_{v,i}^{\lag}\right)^*/\hat{N_i}^* \right\}$ and
 $D_\delta^{\lag}:= \diag\left\{\delta_i^{\lag}\right\}$. 
 The resulting fecundity and transfer matrices are %

\begin{equation}
\label{eqn:Flag}
F^{\lag} = \begin{pmatrix}
0 & P^{\T} D_{\beta} \\
D_{\beta_v}^{\lag}P & 0
\end{pmatrix}     
\end{equation}
and
\begin{equation}
\label{eqn:Vlag}
V^{\lag} = \begin{pmatrix}
D_\delta^{\lag} & 0 \\
0 & G_v
\end{pmatrix}\,,
\end{equation}
with basic reproduction number
\begin{equation}\label{er01}
\left(\mathcal{R}_0^{\lag}\right)^2 = \rho\left(\left(F^{\lag}\right)\left(V^{\lag}\right)^{-1}\right) = \rho\left(P^{\T}D_{\beta}G_v^{-1}D_{\beta_v}^{\lag}P \left(D_\delta^{\lag}\right)^{-1}\right)\, . 
\end{equation}
Details are provided in Appendix \ref{appendixr0}. The next generation matrix is the product of two terms: $P^{\T}D_\beta G_v^{-1}$ corresponding to secondary host infections created by infectious vectors, and $D_{\beta_v}^{\lag} P \left(D_\delta^{\lag}\right)^{-1}$ corresponding to secondary vector infections created by infectious hosts.  Note that for the Lagrangian model, host movement influences the next generation matrix via the mixing matrix $P$ appearing in the fecundity matrix $F^{\lag}$.

%%%%%%%%%%%%%%%%%%%%%%%%%%%%%%%%%%%%%%%%%%%
\subsection{Eulerian for host, Eulerian for vector}\label{sect:eul}
%%%%%%%%%%%%%%%%%%%%%%%%%%%%%%%%%%%%%%%%%%%

Consider now the setting where spatial locations are coupled via migration of both host and vector. The abundances of susceptible, infectious and recovered hosts in patch $i$  are $S_i$, $I_i$ and $R_i$, and the number of susceptible and infectious vectors are $S_{v,i}$ and $I_{v,i}$ respectively. The per capita movement rate of hosts in state $\xx$ from patch $j$ to patch $i$ is $m_{ij}^{\xx}$, for $\xx\in\{\Su,\In,\re\}$, and these rates are recorded in the host movement adjacency matrix $M^{\xx}=\left(m_{ij}^{\xx}\right)_{i,j\leq n}$. The assumption that the movement rate between two patches depends on the state $\xx\in\{\Su,\In,\re\}$ has been considered in studies such as \cite{hsieh2007impact}.  Similarly, the per capita movement rate of vectors from patch $j$ to patch $i$ is $m_{ij}^v$ and these rates are collected in the vector movement adjacency matrix $M^v=\left(m_{ij}^v\right)_{i,j\leq n}$. As before, the host and vector transmission rates for patch $i$ are $\beta_i$ and $\beta_{v,i}$. 
We treat these rates as intrinsic to the patch, and thus take them to be the same as in model (\ref{emodel1}). The recruitment, mortality and recovery rate for hosts in patch $i$ are $\Lambda_i^{\eul}, \mu_i^{\eul} $ and $\gamma_i^{\eul}$ respectively. The recruitment and mortality rate for vectors in patch $i$ are $\Lambda_{v,i}$ and $\mu_{v,i}$ respectively, taken as the same as in model (\ref{emodel1}). We also define the removal rates $\delta_i^{\eul}= \mu_i^{\eul}+\gamma_i^{\eul}$ and $\delta_{v,i} = \mu_{v,i}$.  This comprises a modeling framework where an Eulerian approach is used to model both host and vector movement.  
The corresponding equations are shown in \eqref{emodel2}:

\begin{equation}\label{emodel2}
\begin{array}{ccl}
\dot{S_i} & = & \Lambda_i^{\eul} - \beta_i \frac{S_i}{N_i}I_{v,i} + \sum_{j=1}^n m_{ij}^{\Su} S_{j} - \sum_{j=1}^n m_{ji}^{\Su} S_{i} - \mu_i^{\eul} S_i \\
\dot{I_i} & = & \beta_i \frac{S_i}{N_i}I_{v,i} + \sum_{j=1}^n m_{ij}^{\In} I_{j} - \sum_{j=1}^n m_{ji}^{\In} I_{i}  - \left(\gamma_i^{\eul} + \mu_i^{\eul}\right)I_i \\
\dot{R_i} & = & \gamma_i^{\eul} I_i  + \sum_{j=1}^n m_{ij}^{\re} R_{j} - \sum_{j=1}^n m_{ji}^{\re} R_{i} - \mu_i^{\eul} R_i \\
\dot{S}_{v,i} & = & \Lambda_{v,i} -\beta_{v,i} \frac{I_i}{N_i} S_{v,i} + \sum_{j=1}^n m_{ij}^v S_{v,j} - \sum_{j=1}^n m_{ji}^v S_{v,i} - \mu_{v,i} S_{v,i} \\
\dot{I}_{v,i} & = & \beta_{v,i} \frac{I_i}{N_i} S_{v,i} + \sum_{j=1}^n m_{ij}^v I_{v,j} - \sum_{j=1}^n m_{ji}^v I_{v,i} - \mu_{v,i} I_{v,i}\,,
\end{array}
\end{equation}
for $i=1,\dots,n$.

A summary of the parameters of system (\ref{emodel2}) is given in Table \ref{param1}.  
We assume A\ref{a:M}-A\ref{a:pos} hold for system \eqref{emodel2}.  In addition, we assume A\ref{a:m_host}-A\ref{a:pos2}:

\begin{list}{{\bf A\arabic{acounter}:~}}{\usecounter{acounter}}{\setcounter{acounter}{4}}
    \item The adjacency matrix for host movement has zeros on the diagonal (i.e. $m^{\xx}_{ii} = 0$ for all $i$ and $\xx\in\{\Su,\In,\re\}$).
    \label{a:m_host}
    \item The parameters $\beta_i, \beta_{v,i}, \Lambda_{v,i}, \mu_{v,i}$ are intrinsic to the patch $i$, so they are considered to be the same as in system (\ref{emodel1}).
    \label{a:samepars}
    \item The parameters $\Lambda_i^{\eul},\Lambda_{v,i},\gamma_i^{\eul}$, $\mu_i^{\eul}$ and $\mu_{v,i}$ are all positive.
        \label{a:pos2}
\end{list}

Note that in system \eqref{emodel2}, individuals assume the characteristics of the patch they are currently located in.  For example, a host individual that migrates from $j$ to $i$ now recovers from infection at rate $\gamma_i^{\eul}$.  
The force of infection for patch $i$ in system \eqref{emodel2} depends only upon population abundances in patch $i$ (in contrast to the Lagrangian system \eqref{emodel1}, where the force of infection in $i$ involves contributions from other patches weighted by the mixing matrix $P$).

Let $L^{\xx},L_v$ denote the graph Laplacians corresponding to the host ($M^{\xx}$, for $\xx\in\{\Su,\In,\re\}$) and vector ($M^v$) movement matrices, respectively.  Let $G_v:= L_v+D_{\delta_v}$ as in Section \ref{sect:lag}, and let $D_\mu^{\eul} = \diag\left\{\mu_i^{\eul}\right\}$.  Then at the DFE, the susceptible compartments of  system (\ref{emodel2}) are given by
\begin{equation}\label{eqn:DFEeulerian}
\begin{array}{rcl}
\left(S^{\eul}\right)^*=\left(N^{\eul}\right)^* &=&
\left(L^{\Su} + D_\mu^{\eul}\right)^{-1}\Lambda^{\eul} \, , \\ \left(S_v^{\eul}\right)^*=\left(N_v^{\eul}\right)^* &=& G_v^{-1}\Lambda_v \, ,
\end{array}
\end{equation}
where $\left(S^{\eul}\right)^*, \left(N^{\eul}\right)^*,\left(S_v^{\eul}\right)^*, \left(N_v^{\eul}\right)^*, \Lambda^{\eul}$ and $\Lambda_v $ are column vectors as before. Notice that $\left(N_v^{\eul}\right)^* = \left(N_v^{\lag}\right)^* = G_v^{-1}\Lambda_v$, so we define the vector $N_v^*:=\left(N_v^{\eul}\right)^* = \left(N_v^{\lag}\right)^*$ with entries $N_{v,i}^*= \left(N_{v,i}^{\eul}\right)^* = \left(N_{v,i}^{\lag}\right)^*$.

Consider the basic reproduction number $\R_0^{\eul} $ for system \eqref{emodel2}, computed using the next generation matrix approach.
Define   $D_{\beta}:= \diag\{\beta_i\}$ and $D_{\beta_v}^{\eul} := \diag\left\{\beta_{v,i} \left(N_{v,i}^{\eul}\right)^*/\left(N_i^{\eul}\right)^* \right\}$.  Let $D_\delta^{\eul} := \diag \lbrace \delta_i^{\eul} \rbrace$, and $G := L^{\In} + D_\delta^{\eul}.$  Note that assumption A\ref{a:pos} implies that $G$ is a non-singular $M$-matrix. We show in Appendix \ref{appendixr0} that if 
\begin{equation}
    \label{eqn:Feul}
      F^{\eul} = \begin{pmatrix}
    0 & D_\beta \\
    D_{\beta_v}^{\eul} & 0
    \end{pmatrix}
\end{equation}
and
\begin{equation}
    \label{eqn:Veul}
    V^{\eul} = \begin{pmatrix}
    G & 0\\
    0 & G_v \end{pmatrix}\,,
    \end{equation}
then
\begin{equation}\label{er02}
\left(\mathcal{R}_0^{\eul}\right)^2=\rho\left(\left(F^{\eul}\right)\left(V^{\eul}\right)^{-1}\right) = \rho\left(D_\beta G_v^{-1}D_{\beta_v}^{\eul}G^{-1}\right)\,.
\end{equation}

As for the Lagrangian model, the next generation matrix is the product of two terms, one ($D_\beta G_v^{-1}$) corresponding to secondary host infections created by infectious vectors, and the other ($D_{\beta_v}^{\eul} G^{-1}$) corresponding to secondary vector infections created by infectious hosts.  Note that for the Eulerian model, host movement affects the next generation matrix via the transfer matrix $V^{\eul}$.  This is in contrast with the Lagrangian model, where host movement appears in the fecundity matrix $F^{\lag}$. Additionally, in the Lagrangian model the host transmission rates in the next generation matrix are scaled by host movement [i.e. the $P^{\T}D_\beta G_v^{-1}$ term in \eqref{er01}], whereas in the Eulerian approach they are not [i.e. the $D_\beta G_v^{-1}$ term in \eqref{er02}].  As we will see in the next sections, these differences lead to different values of $\mathcal{R}_0^{\lag}$ and $\mathcal{R}_0^{\eul}$ (in particular, see the end of Appendix \ref{appendixhomogeneous} for intuition on the difference between $\mathcal{R}_0^{\lag}$ and $\mathcal{R}_0^{\eul}$ in a two-patch example).

%%%%%%%%%%%%%%%%%%%%%%%%%%%%%%%%%%%%%%%%%%%
\section{Model comparison through a fundamental matrix} \label{SectionPL}
%%%%%%%%%%%%%%%%%%%%%%%%%%%%%%%%%%%%%%%%%%%

%%%%%%%%%%%%%%%%%%%%%%%%%%%%%%%%%%%%%%%%%%%
\subsection{Preliminaries}
\label{sec:prelims}
%%%%%%%%%%%%%%%%%%%%%%%%%%%%%%%%%%%%%%%%%%%

Our objective is to compare the Lagrangian \eqref{emodel1} and Eulerian \eqref{emodel2} frameworks.  As pointed out by Cosner et al. (see Section 2.2.3 of \cite{cosner2009effects}), the frameworks are in general distinct if we try to relate the number of individuals of both systems. 
Specifically, Cosner et al. showed that the dynamical system resulting from modeling the number of individuals currently located in each patch under the Lagrangian framework does not correspond to an Eulerian model.
Here, we take an alternative approach: we 
compare systems (\ref{emodel1}) and (\ref{emodel2}) 
by tuning the host mobility matrix $M$ so  
that the expected amount of time spent in one location starting from another matches between the two frameworks as closely as possible.

We begin with the analysis of Cosner et al. \cite{cosner2009effects}, who considered when the equilibria between an Eulerian and Lagrangian model can be matched.  Let
$X$ denote a host population type in $\{S,I,N\}$. To match population sizes, we should have that the population size $X_i^{\eul}$ in any patch $i$ for the Eulerian model is equal to the total combined proportions of number of residents $X_j^{\lag}$ of any other patch $j$ that are in $i$ for the Lagrangian model, which is $\sum_j p_{ij}X_j^{\lag}$. In other words, if $X^{\eul}$ and $X^{\lag}$ are column vectors with entries $X_i^{\eul}$ and $X_i^{\lag}$ respectively, we should have:

\begin{equation}\label{xeulxlag}
X^{\eul} = P X^{\lag} \,.
\end{equation}

We will assume in the remainder of the paper that \eqref{xeulxlag} holds at the DFE.  Using 
$\left(N^{\eul}\right)^* = \left(L^{\Su}+D_\mu^{\eul}\right)^{-1} \Lambda^{\eul}$ and    $ \left(N^{\lag}\right)^* = \left(D_{\mu}^{\lag}\right)^{-1} \Lambda^{\lag} $, we set the following assumption: 

\begin{list}{{\bf A\arabic{acounter}:~}}{\usecounter{acounter}}{\setcounter{acounter}{7}}
    \item We assume that the DFEs of  systems (\ref{emodel1}) and (\ref{emodel2}) match in the sense of (\ref{xeulxlag}), so $\left(N^{\eul}\right)^*=P\left(N^{\lag}\right)^*$.
    Specifically, let 
    $\Lambda^{\eul}$ and $\Lambda^{\lag}$ be such that 
    \begin{equation}\label{conditionlambda}
    \left(L^{\Su}+D_\mu^{\eul}\right)^{-1} \Lambda^{\eul} = P \left(D_{\mu}^{\lag}\right)^{-1} \Lambda^{\lag}\,. \end{equation}
    \label{a:DFEs}
\end{list}

The condition in (\ref{conditionlambda}) states that at the DFE the amount of susceptible individual in patch $i$ (Eulerian equilibrium $\left(N_i^{\eul}\right)^*$) is the average (weighted by the row $i$ of $P$) over all different patches $j$ of the number of susceptible individuals that are residents of $j$ (Lagrangian equilibrium $\left(N_j^{\lag}\right)^*$). As pointed out by Cosner et al. \cite{cosner2009effects}, in general it may not be possible to satisfy equation \eqref{conditionlambda}. Namely, for a given $\Lambda^{\lag}$ we may find that $\left(L^{\Su}+D_\mu^{\eul}\right)P\left(D_\mu^{\lag}\right)^{-1} \Lambda^{\lag}$ has some negative entries, and since the equation $\Lambda^{eul} = \left(L^{\Su}+D_\mu^{\eul}\right)P\left(D_\mu^{\lag}\right)^{-1} \Lambda^{\lag}$ is equivalent to (\ref{conditionlambda}), then (\ref{conditionlambda}) would not be true for $\Lambda^{\eul}$ with positive entries. 
However, since $L^{\Su} + D_\mu^{\eul}$ is an $M$-matrix (see \cite{berman1994nonnegative}, page 137), then $L^{\Su}+D_\mu^{\eul}$ is semi-positive, i.e., there exists $x>0$ such that $\left(L^{\Su}+D_\mu^{\eul}\right)x>0$ (see \cite{berman1994nonnegative}, page 136). Therefore, equation (\ref{conditionlambda}) can be satisfied for appropriate $\Lambda^{\lag}, \Lambda^{\eul}$.

In addition, since $D_{\beta_v}^{\lag} := \diag\left\{\beta_{v,i} \left(N_{v,i}^{\lag}\right)^*/\hat{N_i}^* \right\}$, $D_{\beta_v}^{\eul} := \diag\left\{\beta_{v,i} \left(N_{v,i}^{\eul}\right)^*/\left(N_i^{\eul}\right)^* \right\}$, $\left(N_{v,i}^{\lag}\right)^*=\left(N_{v,i}^{\eul}\right)^*$ and  $ \hat{N}^* = P\left(N^{\lag}\right)^* = \left(N^{\eul}\right)^* $ (under (\ref{conditionlambda})), then
$$ D_{\beta_v}:= D_{\beta_v}^{\lag} =D_{\beta_v}^{\eul}\, . $$

Therefore, under A\ref{a:DFEs} we have 

\begin{align} \label{er022}
\left(\mathcal{R}_0^{\lag}\right)^2 & = \rho\left(P^{\T}D_{\beta}G_v^{-1}D_{\beta_v}P \left(D_\delta^{\lag}\right)^{-1}\right)\,, \notag \\ 
\left(\mathcal{R}_0^{\eul}\right)^2 & =\rho\left(D_\beta G_v^{-1}D_{\beta_v}G^{-1}\right) \, .
\end{align}

We can also compare the recovery and mortality rates of both models. 
In general, these parameters are distinct between the two modeling frameworks.  Consider, for example, the mortality rate for patch $i$ in each of the frameworks.  In the Eulerian model, $\mu_i^{\eul}$ reflects only the characteristics of location $i$.  By contrast, the mortality rate $\mu_i^{\lag}$ in the Lagrangian model reflects characteristics of all the spatial locations, weighted according to the proportion of time that a resident of $i$ spends in each location.  Thus,  
we model this relationship by averaging the rates $\mu_j^{\eul}$ and setting $\mu_i^{\lag} = \sum_{j} p_{ji}\mu_j^{\eul}$.  
These considerations lead us to the following assumption:

\begin{list}{{\bf A\arabic{acounter}:~}}{\usecounter{acounter}}{\setcounter{acounter}{8}}
    \item We assume that the parameters $\mu_i^{\lag}$, $\gamma_i^{\lag}$, $\delta_i^{\lag}$ are related to $\mu_i^{\eul}$, $\gamma_i^{\eul}$, $\delta_i^{\eul}$ by
    \begin{equation}\label{eqn:rel_abs}
    \begin{array}{rcl}
     \mu_i^{\lag} &=& \sum_{j} p_{ji}\mu_j^{\eul}\,, \\ %\qquad  
     \gamma_i^{\lag} &=& \sum_{j} p_{ji}\gamma_j^{\eul}\,, \\ %\qquad  
     \delta_i^{\lag} &=& \sum_{j} p_{ji}\delta_j^{\eul}\,.
\end{array}
     \end{equation}
     
    In matrix form, this means $D_\delta^{\lag} = \diag\left\{\mathbf{1}^{\T} D_\delta^{\eul}P\right\}$, where $\mathbf{1} = (1,1,\ldots,1)^{\T}$.
    \label{a:abs}
\end{list}

%%%%%%%%%%%%%%%%%%%%%%%%%%%%%%%%%%%%%%%%%%%
\subsection{Problem formulation} \label{sec:formulation}
%%%%%%%%%%%%%%%%%%%%%%%%%%%%%%%%%%%%%%%%%%%

We compare the Lagrangian and Eulerian modeling frameworks through a relationship that can be interpreted in terms of a fundamental matrix for a Markov process.  Systems \eqref{emodel1} and \eqref{emodel2} are deterministic, not stochastic.  However, aspects of both systems can be interpreted in terms of a fundamental matrix for a continuous time random walk \cite{dobrow2016introduction}.
We will use this fundamental matrix to relate the two systems.  
Consider the next generation matrix $F^{\eul}\left(V^{\eul}\right)^{-1}$ for system \eqref{emodel2} \cite{van2002reproduction}.   The transfer matrix $V^{\eul}$ includes a block $G = L^{\In}+D_\delta^{\eul}$ that generates an absorbing random walk on the host movement network.  $G^{-1}$ corresponds to the fundamental matrix of this random walk, with $(i,j)$ entry giving the expected time that an infectious individual starting in $j$ spends in $i$ before being absorbed (removed) from the system \cite{tavare1979note,van2002reproduction}.  This interpretation of the transfer matrix underlies intuition for $FV^{-1}$ as giving the number of infectious individuals in the `next generation', and corresponding threshold of $\rho\left(FV^{-1}\right) >1$ for disease invasion.  See \cite{van2002reproduction}.

Now consider the mixing matrix $P$ for the Lagrangian system \eqref{emodel1}.  The entries $p_{ij}$ give the probability that a resident of $j$ is in patch $i$, and $1/\delta_j^{\lag}$ gives the expected time that a resident of $j$ stays infectious, so $p_{ij}/\delta_j^{\lag}$ represents the expected time that an infectious host from patch $j$ spends in patch $i$ according to the Lagrangian approach. 
Matching the expected `residence times' (times spent in $i$, starting from $j$) for the Eulerian and Lagrangian frameworks and applying A\ref{a:abs}, we have:
    \begin{equation}\label{relationPL}
    \left(L^{\In}+D_\delta^{\eul}\right)^{-1} = P\left(D_\delta^{\lag}\right)^{-1} = (P) \left(\diag^{-1}\left\{\mathbf{1}^{\T} D_\delta^{\eul}P\right\}\right) \,.
    \end{equation}
    \label{a:fundamental}
    
As we will see in Section \ref{sectioninconsistent}, for a given $D_\delta^{\eul}$ and $P$ it may not always be possible to find a graph Laplacian matrix $L^{\In}$ such that (\ref{relationPL}) holds. This fact leads us to the following definition.

\begin{definition}{{\bf (Consistency)}}
Assume A\ref{a:M}-A\ref{a:abs}  and suppose that $P$ and $D_\delta^{\eul}$ are given. We say that systems \eqref{emodel1} and \eqref{emodel2} are {\bf consistent} if there exists a graph Laplacian matrix $L^{\In}$  satisfying (\ref{relationPL}), and {\bf inconsistent} if such a matrix does not exist. \label{def:consistency}
\end{definition}

Note that as $L^{\In}+D_\delta^{\eul}$ and $D_\delta^{\lag}$ are non-singular, if  systems \eqref{emodel1} and \eqref{emodel2} are consistent, 
then $P$ is non-singular.  However, as we will see in Section \ref{sectioninconsistent}, the converse is not necessarily true. We will assume in the following sections that $P$ is non-singular. 

\begin{list}{{\bf A\arabic{acounter}:~}}{\usecounter{acounter}}{\setcounter{acounter}{9}}
    \item $P$ is non-singular.
    \label{a:Pinvertible}
\end{list}

Notice that under A\ref{a:Pinvertible}, relationship (\ref{relationPL}) is equivalent to 
 
\begin{equation}
L^{\In} = D_\delta^{\lag}P^{-1}- D_\delta^{\eul}\,. 
\end{equation}
 
Let $p_{ij}'$ denote the $(i,j)$ entry of the inverse of $P$, that is, let $P^{-1} = (p_{ij}')_{i,j\leq n}$.
Then the $j^{th}$ column sum of $D_\delta^{\lag}P^{-1}- D_\delta^{\eul}$ is

\begin{equation}\label{eqn:sum_col_zero}
\left(\sum_i p_{ij}' \delta_i^{\lag}\right)-\delta_j^{\eul} = \left(\sum_i p_{ij}' \sum_k p_{ki}\delta_k^{\eul}\right)-\delta_j^{\eul} = \left(\sum_k  \delta_k^{\eul} \sum_{i}p_{ki}p_{ij}'\right)-\delta_j^{\eul}=0 \,.   
\end{equation}
 
The above condition is consistent with the graph Laplacian having zero column sums. However, we also want the off-diagonal elements of $D_\delta^{\lag}P^{-1}- D_\delta^{\eul}$ to be non-positive, which is not in general the case as we will show in the example of Section \ref{sectioninconsistent}.
In consequence, we are interested in  finding
\begin{equation}
\label{eqn:E_min}
E =\inf_{L \in \mathbb{R}^{n \times n}} \left\{ \norm{L - (D_\delta^{\lag}P^{-1}- D_\delta^{\eul})}_F : L \text{ is a graph Laplacian matrix }\right\},
\end{equation}
where we use the Frobenius norm defined by $\norm{B}_F = \sqrt{\tr(BB^{\T})}$ for a given matrix $B$. The Frobenius norm
allows us to treat \eqref{eqn:E_min} as a non-negative least squares problem.
For example, in the 2 by 2 case, we have that
\begin{equation}
    E = \inf_{m\geq 0}  \norm{L^* m - \bar{x}}_2 \,, 
\end{equation}
where 

$$L^* = \begin{pmatrix}
1 & 0\\
-1 & 0\\
0 & 1 \\
0 & -1
\end{pmatrix},\quad m = \begin{pmatrix}
m_{21}\\ m_{12}
\end{pmatrix}, \quad D_\delta^{\lag}P^{-1}- D_\delta^{\eul} = \begin{pmatrix}
a & c \\ b & d
\end{pmatrix}, \quad \bar{x} =(a, b, c, d)^{\T}\,. $$

Notice that $ \norm{m}_2 \leq \norm{L}_F = \norm{L^* m}_2  \leq \norm{L^* m -\bar{x}}_2 + \norm{\bar{x}}_2 $. Therefore, if $\{m_k\}_{k\geq 1}$ is a sequence such that $\lim_{k \rightarrow \infty} \norm{L^*m_k-\bar{x}}_2=E$, then $\limsup_{k\geq 1} \norm{m_k}_2 \leq E + \norm{\bar{x}}_2$. In consequence,  $\{\norm{m_k}_2\}_{k\geq 1}$ is bounded and therefore $\inf_{m\geq 0}  \norm{L^* m - \bar{x}}_2$ is attained, i.e., $E=\min_{m\geq 0}  \norm{L^* m - \bar{x}}_2.$

In general, $L^*$ is an $n^2 \times n(n-1)$ matrix, and from the Karush-Kuhn-Tucker  conditions the optimum $\bar{m}$ satisfies $ (L^* \bar{m} - \bar{x})^{\T} L^* \bar{m}=0, \bar{m} \geq 0$ (see Section 10.10 of \cite{byrne2014first}), giving 

\begin{equation}
\label{eqn:kkt}
\bar{x}^{\T}L^*\bar{m} = \norm{L^*\bar{m}}_2^2\geq 0, 
E^2 =(L^* \bar{m} - \bar{x})^{\T} (L^* \bar{m} - \bar{x})
= \norm{D_\delta^{\lag}P^{-1}- D_\delta^{\eul} }_F^2 - \bar{x}^{\T}L^*\bar{m}\,. 
\end{equation}

Hence we have the following upper bound for $E$:
\begin{equation}\label{boundconsistency}
 E  \leq \norm{ D_\delta^{\lag}P^{-1}- D_\delta^{\eul} }_F\,.
\end{equation}

%%%%%%%%%%%%%%%%%%%%%%%%%%%%%%%%%%%%%%%%%%%
\subsection{Consistency} \label{sec:consistency}
%%%%%%%%%%%%%%%%%%%%%%%%%%%%%%%%%%%%%%%%%%%

%%%%%%%%%%%%%%%%%%%%%%%%%%%%%%%%%%%%%%%%%%%
\subsubsection{Sufficient condition} \label{Ssufficientcondition}
%%%%%%%%%%%%%%%%%%%%%%%%%%%%%%%%%%%%%%%%%%%

By direct calculation we can show that the condition $p_{12}+p_{21} < 1$ guarantees that systems (\ref{emodel1}) and (\ref{emodel2}) are consistent in the two-patch setting. This suggests that the off-diagonal terms of the mixing matrix $P$ must be sufficiently small for the Lagrangian and Eulerian systems to be consistent. On the other hand, if 
\begin{equation}\label{eqn:examplesufficient}
 P=\begin{pmatrix}
  9/10  &  0  &  0 \\
    0 &   1  &  1/10 \\
    1/10  &  0 &   9/10
\end{pmatrix}, \text{ then }D_\delta^{\lag}P^{-1}-D_\delta^{\eul} = \begin{pmatrix} 10/9 \delta_1^{\lag} -\delta_1^{\eul} & 0 & 0 \\ 1/81\delta_2^{\lag} & \delta_2^{\lag} - \delta_2^{\eul} & -1/9\delta_2^{\lag} \\ -10/81\delta_3^{\lag} & 0 & 10/9\delta_3^{\lag} -\delta_3^{\eul} \end{pmatrix}.   
\end{equation}
From (\ref{eqn:examplesufficient}),  $D_\delta^{\lag}P^{-1}-D_\delta^{\eul}$ has positive off-diagonal entries, and thus systems (\ref{emodel1}) and (\ref{emodel2}) are inconsistent. In this example, some entries $p_{ij}$ of $P$ are small (for example $p_{21} = 0$), which suggests that non-diagonal entries of the mixing matrix must also be sufficiently large in order for the Eulerian and Lagrangian frameworks to be consistent.
In  Proposition \ref{prop:sufficient} below, we prove   
that a sufficient condition for the consistency of systems (\ref{emodel1}) and (\ref{emodel2}) is that the off-diagonal entries of the mixing matrix belong to an intermediate range $p_* < p_{ij} < p^*$.

\begin{proposition}\label{prop:sufficient}
Assume A\ref{a:M}-A\ref{a:Pinvertible} and let $P = (p_{ij})_{i,j\leq n}$ be the mixing matrix associated with system (\ref{emodel1}), where $n\geq 2$. Let $p_*$ and $p^*$ be constant numbers in the interval $(0,1)$ such that 
\begin{equation}\label{scondition0}
\frac{4(n-1)^2p^*}{1-p^*} < 1  \text{ and } p_* = \frac{4(n-1)^2(p^*)^2}{1-p^*} \,.
\end{equation}
In addition, suppose that 
\begin{equation}\label{scondition}
p_* \leq p_{ij} \leq p^*, \text{ for } i\neq j \,.
\end{equation}
Then, the systems (\ref{emodel1}) and (\ref{emodel2}) are consistent.
\end{proposition}

\begin{proof}
From (\ref{eqn:sum_col_zero}), it follows that the column sums of $D_\delta^{\lag}P^{-1}-D_\delta^{\eul}$ are all zero. Consequently, we just need to show that the off-diagonal entries of $D_\delta^{\lag}P^{-1}-D_\delta^{\eul}$ are non-positive in order to prove that systems (\ref{emodel1}) and (\ref{emodel2}) are consistent. Observe that if $p^*$ goes to zero, then $4(n-1)^2p^*/(1-p^*)$ approaches to zero, so the first part of (\ref{scondition0}) can be satisfied by small enough $p^*$. The condition (\ref{scondition0}) then implies that $$\frac{p_*}{p^*} = \frac{4(n-1)^2p^*}{1-p^*} < 1 \,. $$
Therefore, the interval $[p_*, p^*]$ is non-empty and we can pick $p_{ij}$ such that (\ref{scondition}) holds.
In order to show that the off-diagonal entries of $D_\delta^{\lag}P^{-1}-D_\delta^{\eul}$ are non-positive, it suffices to prove that the off-diagonal entries of $P^{-1}$ are non-positive. Let $\Delta P$ be a matrix such that $P = I-\Delta P$, and consider the matrix 1-norm defined by $\norm{B}_1 =\max_j \sum_i |b_{ij}|$ for a given matrix $B$. We can write $P^{-1} = (I-\Delta P)^{-1} = \mathcal{P}^{(0)} + \mathcal{P}^{(1)}$, where $\mathcal{P}^{(0)} = I+\Delta P$ and $\mathcal{P}^{(1)} = (\Delta P)^2 \sum_{k\geq 0} (\Delta P)^k$. If $i\neq j$, then $\mathcal{P}^{(0)}_{ij} = -p_{ij}$, so we want to show $|\mathcal{P}^{(1)}_{ij}| \leq p_{ij}$ to get that the off-diagonal elements of $P^{-1}$ are non-positive. Indeed, from (\ref{scondition0}) and (\ref{scondition}) it follows that $\norm{\Delta P}_1 \leq 2(n-1)p^*$ and

\begin{align*}
|\mathcal{P}^{(1)}_{ij}|  & \leq  \norm{\mathcal{P}^{(1)}}_1 \\  
& \leq  \norm{\Delta P}_1^2\sum_{k\geq0}\norm{\Delta P}_1^k \\ & =  \norm{\Delta P}_1^2/(1-\norm{\Delta P}_1) \\ & \leq \left[2(n-1)p^*\right]^2/(1-p^*) \\
& = p_* \leq p_{ij} \,,  
\end{align*}
as we desired.    
\end{proof}

Proposition \ref{prop:sufficient} establishes that systems \eqref{emodel1} and \eqref{emodel2} are consistent when the off-diagonal entries of the mixing matrix $P$ lie in an interval $\lbrack p_*, p^* \rbrack \subset (0,1)$.  The upper and lower bounds for this interval satisfy \eqref{scondition0} and \eqref{scondition}.  We note that these bounds are not unique, as more than one $p^*$ can satisfy \eqref{scondition0} and \eqref{scondition}.  
Furthermore, the width of the resulting interval $(p_*,p^*)$ may be small for large $n$.

For example, let $n=3$.  From \eqref{scondition0} it suffices to choose $p^*$ such that $p^*<1/17$. Let us consider, for instance, $p^*=0.05$. Consequently, from \eqref{scondition0} we have $p_* = 0.0421$  and then 
$p_* = 0.0239 \leq p_{ij} \leq p^* = 0.1324$ is a sufficient condition for consistency of systems \eqref{emodel1} and \eqref{emodel2}. We note that this interval may not be the widest interval among those obtained using other values of $p^*$ satisfying \eqref{scondition0}.

Additionally, (\ref{scondition0}) may be improved. Namely, we can write $P^{-1} = \mathcal{P}^{(0)}+\mathcal{P}^{(1)}$ where $\mathcal{P}^{(0)} = I+\Delta P + (\Delta P)^2$ and $\mathcal{P}^{(1)} = (\Delta P)^3 \sum_{k\geq 0} (\Delta P)^k$.  By imposing the condition $p_* \leq p_{ij}\leq p^*$, with $i\neq j$, we have that $p_*$ and $p^*$ must satisfy 

\begin{align}\label{eqn:conditiondetP}
-\mathcal{P}^{(0)}_{ij} & \geq p_* + 2(n-1)p_*^2 - (n-2)(p^*)^2 \notag \\ 
& \geq \frac{\left[2(n-1)p^*\right]^3}{1-p^*} \notag \\ 
& \geq |\mathcal{P}^{(1)}_{ij}| \,.
\end{align}

Figure \ref{fig:condition3cycles} shows the region of pairs $(p_*,p^*)$ that satisfy the inequality (\ref{eqn:conditiondetP}) for $n=5$. Point A of Figure \ref{fig:condition3cycles} indicates that  $p_* = 0.0065 \leq p_{ij} \leq p^* = 0.022$, for $i\neq j$, is a sufficient condition for consistency of systems \eqref{emodel1} and \eqref{emodel2}. This condition improves (\ref{scondition0}), where $p^* = 0.015$ is associated to $p_* = 0.0065$ and the interval $(0.0065,0.015)$, which is smaller than the interval $(0.0065,0.022)$ corresponding to (\ref{eqn:conditiondetP}).

\begin{figure}[H]
\centering
\includegraphics[scale= 0.6]{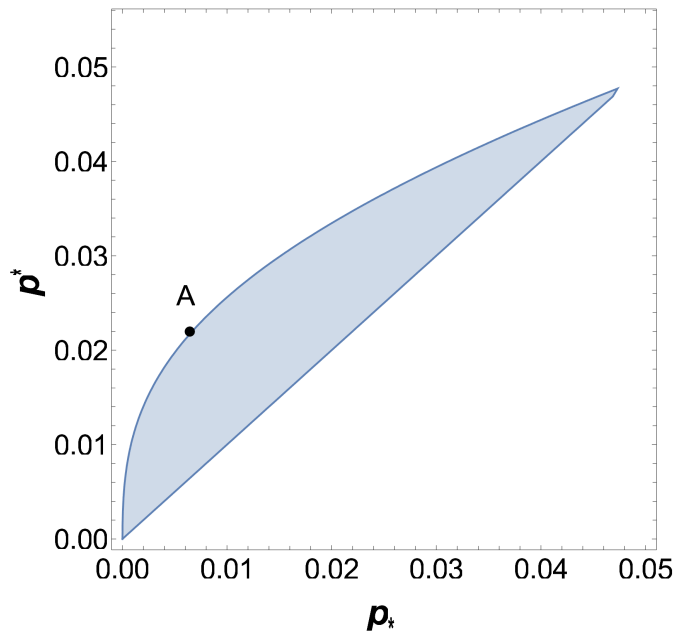}
\caption{Inequality (\ref{eqn:conditiondetP}) holds in the blue region for $n=5$. For every point $(p_*,p^*)$ on the top boundary of the blue region we have a consistency condition. Namely, under A\ref{a:M}-A\ref{a:Pinvertible}, if $p_{ij} \in \lbrack p_* , p^* \rbrack $ for all $i\neq j$, then systems (\ref{emodel1}) and (\ref{emodel2}) are consistent. For example, at the black point A we have $p_* = 0.0065$ and $p^*=0.022$. Therefore, if $0.0065 \leq p_{ij}\leq 0.022$ for $i\neq j$, then the systems (\ref{emodel1}) and (\ref{emodel2}) are consistent.}
\label{fig:condition3cycles}
\end{figure}

%%%%%%%%%%%%%%%%%%%%%%%%%%%%%%%%%%%%%%%%%%%
\subsubsection{A consistent example: star graphs} \label{Sconsistentexample}
%%%%%%%%%%%%%%%%%%%%%%%%%%%%%%%%%%%%%%%%%%%

In this section we give an example where the Eulerian and Lagrangian systems are consistent.  Specifically, we consider a setting where the mixing matrix $P$ corresponds to a star graph, where the `hub' node is the only location that residents of other patches visit.  This setting is motivated by empirical networks where there exist nodes $k$ for which all the $p_{ki}$ are large.  For example, in the data analyzed in \cite{ruktanonchai2016identifying} on malaria in Namibia, 
non-residents are much more likely to visit a few locations (e.g. the capital Windhoek) than others.  A schematic of the class of mixing matrices considered in this section is shown in Figure \ref{fig:stargraph}.  We will show that for such $P$, systems \eqref{emodel1} and \eqref{emodel2} are consistent.

\begin{figure}[H]
    \centering
    \includegraphics[scale = 0.6]{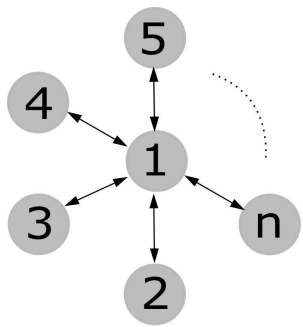}
    \caption{Star graph as example of consistency. }
    \label{fig:stargraph}
\end{figure}

We will use the Sherman-Morrison formula
\begin{equation}\label{eqn:sherman}
\left(A + vu^{\T}\right)^{-1} = A^{-1} - \frac{A^{-1}vu^{\T}A^{-1}}{1 + u^{\T}A^{-1}v} \,, 
\end{equation}
where $A$ is non-singular and $v,u$ are column vectors such that $1 + u^{\T} A^{-1} v \neq 0$.
Let us suppose that 
\begin{equation} \label{nullpij}
P = \diag\{1,1-p_{12}, \ldots, 1-p_{1n}\} + (1,0,\ldots,0)^{\T}(0,p_{12},\ldots,p_{1n})\,, 
\end{equation}
$D_\delta^{\lag} = D_\delta^{\eul} = D_\delta = \delta I$ and let $b_{ij}$ denote the $(i,j)$ entry of $D_\delta^{\lag}P^{-1}- D_\delta^{\eul}$. 
Using the Sherman-Morrison formula (\ref{eqn:sherman}), we get 
$$P^{-1} = \diag\{1,1/(1-p_{12}), \ldots, 1/(1-p_{1n})\} - (1,0,\ldots,0)^{\T}(0,p_{12}/(1-p_{12}),\ldots,p_{1n}/(1-p_{1n}))$$ and 
$$ b_{ij} = \frac{- \delta p_{ij}}{1 - p_{ij}}, \textrm{ for } i\neq j\,.$$

Then, as $p_{ij} < 1$ for $i \neq j$, systems \eqref{emodel1} and \eqref{emodel2} are consistent.

Moreover, let $f: \mathbb{R}^{n\times n} \rightarrow \mathbb{R}^{n\times n}$ be defined by $f(P):=  D_\delta(P^{-1}-I)$, and let $U$ denote the open set $U = \{M \in \mathbb{R}^{n\times n}: \text{ the off-diagonal entries of }M \text{ are negative }\}$. If $P_0$ is of the form (\ref{nullpij}), then $f(P_0) \in U$. 

By continuity of $f$, for $P$ with small enough $p_{ij}, i \neq j$ and $i\neq 1$, we have that $f(P) \in U$, 
i.e, systems (\ref{emodel1}) and (\ref{emodel2}) are consistent.

%%%%%%%%%%%%%%%%%%%%%%%%%%%%%%%%%%%%%%%%%%%
\subsection{An inconsistent example}\label{sectioninconsistent}
%%%%%%%%%%%%%%%%%%%%%%%%%%%%%%%%%%%%%%%%%%%

We now present an example where systems \eqref{emodel1} and \eqref{emodel2} are inconsistent, and in fact the upper bound in \eqref{boundconsistency} is attained.

Suppose that

\begin{equation}\label{Pinconsistent}
\begin{array}{ccl}
p_{ij} & := & p_i, i\neq j \,, \\
p_{ii} & :=  & 1- \sum_{k\neq i} p_k \,.
\end{array}
\end{equation}

Let $\theta:= 1- \sum_{k=1}^n p_k \neq 0$. Then we have $$ P =\theta I + (p_1,\ldots, p_n)^{\T} (1, \ldots, 1) \,.$$

As $A := \theta I$ is non-singular and $1+  (1, \ldots, 1)A^{-1} (p_1,\ldots, p_n)^{\T} = 1+ \frac{1}{\theta} \sum_{k=1}^n p_k =\frac{1}{\theta} \neq 0 $, applying the Sherman-Morrison formula gives $$ P^{-1} =\frac{1}{\theta} I - \frac{1}{\theta} (p_1,\ldots, p_n)^{\T} (1, \ldots, 1)\,.$$

In addition, suppose that $\delta_i^{\eul} :=\delta, i=1,\ldots, n$, so that $D_\delta^{\lag} = D_\delta^{\eul} = \delta I$. Therefore, if  $ D_\delta^{\lag}P^{-1}- D_\delta^{\eul} = \delta (P^{-1}-I) := (b_{ij})_{i,j\leq n}$, then we have that
\begin{equation}
\label{eqn:m_sherm}
b_{ij} = -\frac{\delta p_i}{\theta}, \quad i\neq j \,.
\end{equation}
For sufficiently large $p_k$ (such that $\sum_k p_k > 1$) we have $\theta < 0$, leading to positive $b_{ij}$ for $i\neq j$, which is inconsistent with the off-diagonal entries of a Laplacian matrix. In this case the optimum in \eqref{eqn:kkt} is $\bar{m}=0$ and the error is $E = || D_\delta^{\lag}P^{-1}- D_\delta^{\eul}||_F$, which is the largest possible error.

As a specific example, consider three patches with $p_1 = 0.8, p_2 = 0.15, p_3 = 0.15$.  Then $\theta = -0.1 < 0$ and  
\begin{equation}\label{matrixexample}
P=\begin{pmatrix}
  0.7  &  0.8  &  0.8 \\
    0.15 &   0.05  &  0.15 \\
    0.15  &  0.15 &   0.05
\end{pmatrix},
\end{equation}
for which $L=0$ is the solution of (\ref{eqn:E_min}).
Thus it is possible for systems (\ref{emodel1}) and (\ref{emodel2}) to not only be inconsistent, but in fact for the upper bound in \eqref{boundconsistency} to be attained.  We note that the preceding example requires that some of the off-diagonal entries of $P$ are large, which may not be realistic in situations where host individuals spend the majority of their time in a distinguished `home' location.

%%%%%%%%%%%%%%%%%%%%%%%%%%%%%%%%%%%%%%%%%%%
\section{Two-patch network}\label{section2by2}
%%%%%%%%%%%%%%%%%%%%%%%%%%%%%%%%%%%%%%%%%%%

In this section we explore results obtained from consistency and inconsistency of systems (\ref{emodel1}) and (\ref{emodel2}) for two-patch systems ($n=2$). 
In Section \ref{sec:finalsize2patches} we compare the final outbreak size and the basic reproduction number obtained from systems (\ref{emodel1}) and (\ref{emodel2})  for an example where the systems are inconsistent. In Section \ref{sec:comparisonhomo} we state  explicit bounds for the relative difference between the basic reproduction numbers of systems (\ref{emodel1}) and (\ref{emodel2}) when the transmission and removal rates are the same for both patches in Proposition \ref{prep:homogeneous}. In Section \ref{sec:comparisonhomo} we also compare the basic reproduction number of both systems for a particular example where the removal rates for the two patches are different.

%%%%%%%%%%%%%%%%%%%%%%%%%%%%%%%%%%%%%%%%%%%
\subsection{Final outbreak size and basic reproduction number for an inconsistent example}
\label{sec:finalsize2patches}
%%%%%%%%%%%%%%%%%%%%%%%%%%%%%%%%%%%%%%%%%%%

We now consider the functional implications of consistency / inconsistency of the Eulerian and Lagrangian frameworks, in terms of important disease quantities such as the basic reproduction number and final outbreak size.

Let us consider an example for a two-patch network where $P = \begin{pmatrix} 1-p_{21} & p_{12} \\ p_{21} & 1-p_{12} \end{pmatrix} $, with $p_{12}+p_{21}>1$. Let us also assume that $D_\delta^{\lag} = D_\delta^{\eul} = D_\delta = \delta I$, $D_{\beta_v} = \beta_v I$, $D_\beta = \beta I$, $D_{\delta_v} = \delta_v I$ and $L_v = m_v\begin{pmatrix} 1 & -1 \\ -1 & 1 \end{pmatrix}$.  Then the Eulerian and Lagrangian frameworks are inconsistent by the argument in Section \ref{sectioninconsistent}, and $L=0$ optimizes \eqref{eqn:E_min} in this case, meaning that the solution to \eqref{eqn:E_min} corresponds to a completely disconnected set of nodes in the Eulerian framework.  Thus the networks in the Lagrangian and Eulerian frameworks are wildly different.  Here we examine how this difference in connectivity corresponds to differences in $\R_0$ and final outbreak size.

In Figure \ref{figfinalsize} we use the parameters $\delta = 1/150$, $\beta = 0.3\times 0.1$,  $\delta_v = 0.05$, $m_v=0.02$ \cite{ruktanonchai2016identifying} and define $p_{12}=0.95$, $p_{21}=0.1$ (therefore $ p_{12}+p_{21}>1$, which implies inconsistency). In Figure \ref{figfinalsize}(a) we observe that the final outbreak size obtained from the Lagrangian system is larger than the final outbreak size obtained from the Eulerian system. Furthermore, for values of $\R_0^{\eul}$ around one, we get a significant  relative difference between the outbreak sizes for the two systems. For example, when $\R_0^{\eul} = 1.2$, the outbreak size of Lagrangian system is $220\%$ larger than the outbreak size of the Eulerian system. 
Thus inconsistency in terms of Definition \ref{def:consistency} might lead to significant differences in outbreak size between the Eulerian and Lagrangian frameworks.
Moreover, we can modify the example of Figure \ref{figfinalsize} so the systems are consistent and there is still a significant difference in outbreak size for intermediate values of $\mathcal{R}_0^{\eul}$. 
On the other hand, when $\R_0^{\eul}=1.6$, the percentage change is $25.2\%$, and when $\R_0^{\eul}=0.5$ the final sizes are almost the same. In Figure \ref{figfinalsize}(b) we observe that $\R_0^{\lag}>\R_0^{\eul}$ if $0\leq \R_0^{\eul} \leq 2$.   Additionally, $\R_0^{\lag}$ is linear with respect to $\R_0^{\eul}$, where $\R_0^{\lag}$ is at most $18.4\%$ larger than $\R_0^{\eul}$.

\begin{figure}[H]
\centering
\subfloat[]{\includegraphics[scale=0.4]{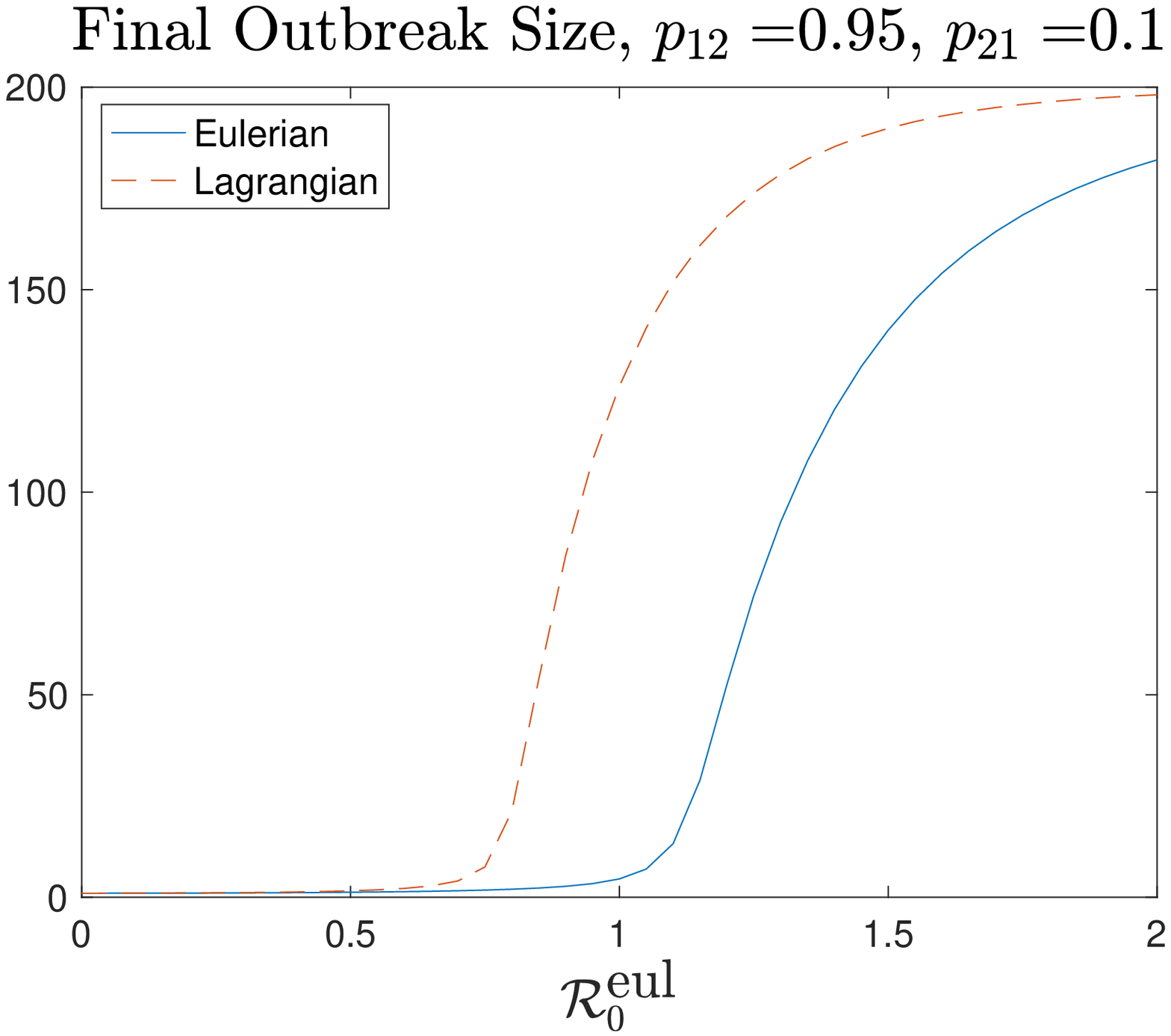}}
\subfloat[]{\includegraphics[scale=0.4]{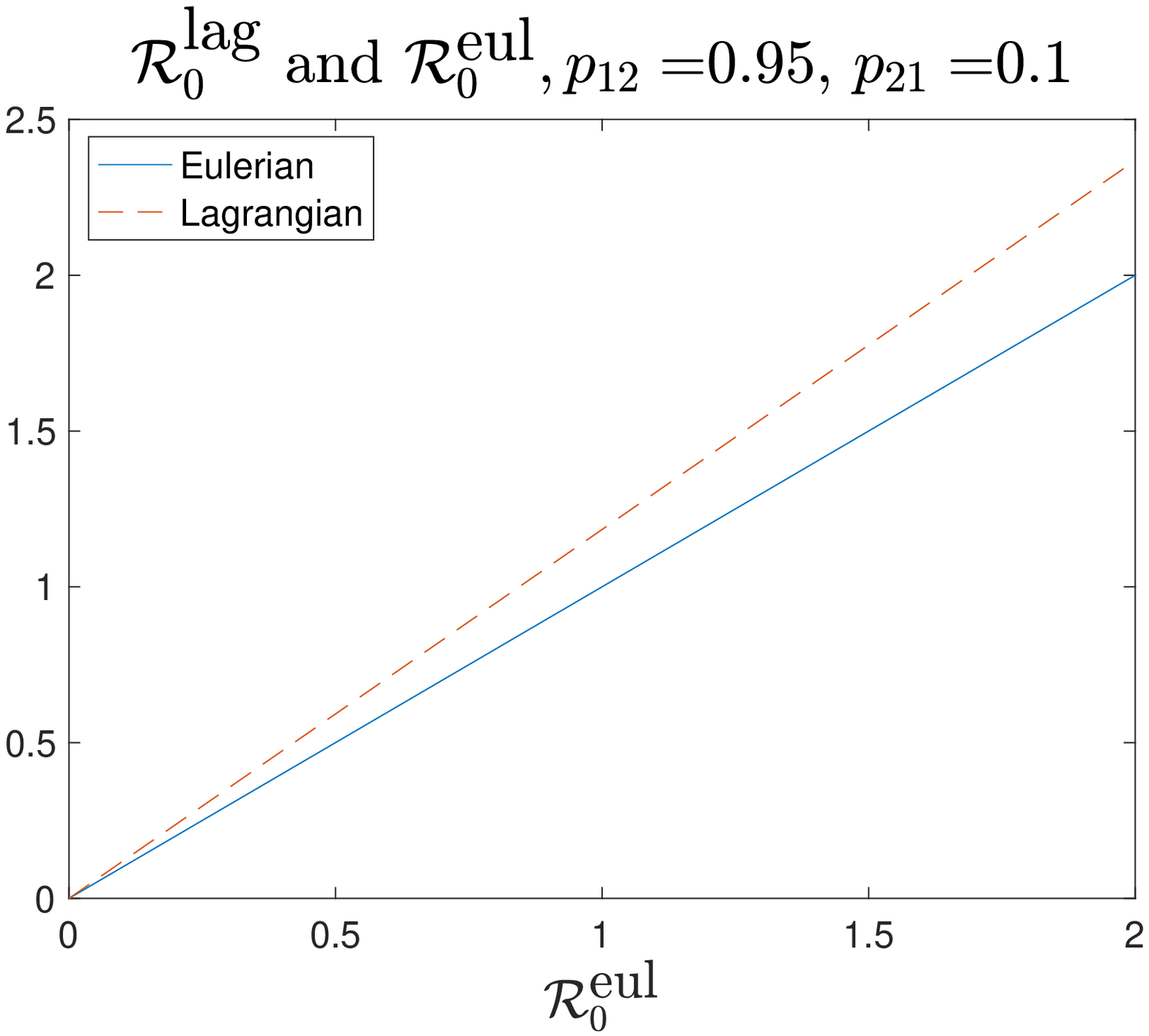}} \\
\caption{In (a) we have the final outbreak size for the inconsistent example given in Section \ref{sec:finalsize2patches} (the recovered individuals $R_1(T)+R_2(T)$ for large enough $T$). In (b) we have the comparison between the basic reproduction number of both models. The used parameters are  $a=0.3, b=0.1, \mu=0,\delta = \gamma = 1/150,\delta_v = 0.05, \beta = ab, m_v=0.02$ \cite{ruktanonchai2016identifying}, and the units are as in Table \ref{param1}. We also assume $G_v= m_v(2I- \mathbf{1} \mathbf{1}^{\T}), N_v=80, \Lambda_v $ such that $G_v^{-1}\Lambda_v = N_v \mathbf{1},  \Lambda^{\lag}= \Lambda^{\eul} = (0,0)^{\T}, N_0 = 100.$ For a given value of $\mathcal{R}_0^{\eul}$ we choose $\beta_{v,1}^{\eul}=\beta_{v,2}^{\eul}$  such that $\beta_v := \beta_{v,i} ^{\eul} =  \left(\mathcal{R}_0^{\eul}\right)^2\delta_v \delta N_0 / (\beta N_v) $. We also define $\beta_{v,1}^{\lag}$ and $\beta_{v,2}^{\lag}$ such that  $D_{\beta_v}^{\lag}=D_{\beta_v}^{\eul} = \beta_v N_v/N_0 I$. The initial conditions are $S_1(0) = N_0, S_2(0) = N_0-1, I_1(0) = 0, I_2(0)=1, R_1(0)=R_2(0)=0, S_{v,1}(0) = S_{v,2}(0) = 50, I_{v,1}=I_{v,2}=0 $ and the final outbreak is taken at time $T=3000$. 
}
\label{figfinalsize}
\end{figure}
%

%%%%%%%%%%%%%%%%%%%%%%%%%%%%%%%%%%%%%%%%%%%
\subsection{Homogeneous infection} \label{sec:comparisonhomo}
%%%%%%%%%%%%%%%%%%%%%%%%%%%%%%%%%%%%%%%%%%%

In this section, we consider the special case where parameter values for transmission and removal are the same across both locations in the two-patch network.  We will show that in this case we can bound the difference between the basic reproduction numbers for the Eulerian and Lagrangian frameworks.  Furthermore, we show that in this `homogeneous infection' setting, the basic reproduction number for the Lagrangian model is greater than or equal to the reproduction number for the Eulerian model.

Consider the case where $0<p_{12},p_{21} <1/2$. Under these conditions, the systems (\ref{emodel1}) and (\ref{emodel2}) are consistent. We assume A\ref{a:M}--A\ref{a:Pinvertible}, and pick the same transmission and recovery parameters for both patches.  We refer to this setting as the homogeneous infection scenario. The following proposition describes the behavior of the quantities $\mathcal{R}_0^{\eul}$ and $\mathcal{R}_0^{\lag}$ for $0<p_{12}<1/2$ in the homogeneous infection scenario. In this proposition, we can interchange $p_{12}$ by $p_{21}$ and obtain symmetric results. A proof and additional details are given in Appendix \ref{appendixhomogeneous}.

\begin{proposition}\label{prep:homogeneous}
Assume A\ref{a:M}--A\ref{a:Pinvertible} and suppose that $D_\beta = \beta I$, $D_{\beta_v} = \beta_v 	I$, $D_\delta^{\eul} = D_\delta^{\lag} = D_\delta := \delta I$, $D_{\delta_v} =\delta_v I$ and $m_{12}^v = m_{21}^v = m_v$. Define $\mathcal{M}:=D_{\beta}G_v^{-1}D_{\beta_v}P \left(D_\delta^{\lag}\right)^{-1}$ and fix $p_{21}$ in the interval $(0,1/2)$. Then $\rho(P^{\T}\mathcal{M})$ is a function of $p_{12}$ where $0<p_{12}<1/2$, and we have that:
\begin{enumerate}[a)]
    \item $\left(\R_0^{\eul}\right)^2 = \rho(\mathcal{M}) = (\beta \beta_v)/(\delta \delta_v)$ is constant on $0<p_{12}<1/2$.
    \item $\left(\R_0^{\lag}\right)^2 = \rho(P^{t}\mathcal{M})$ is decreasing on $(0,p_{21})$, increasing on $(p_{21},1/2)$ and attains its absolute minimum over $[0,1/2]$ with value $\left(\R_0^{\eul}\right)^2 = (\beta \beta_v)/(\delta \delta_v)$ at $p_{12}=p_{21}$.
    \item In addition, we have the inequality 
    
\begin{equation}\label{Bound0}
\begin{split}
\frac{\left(\R_0^{\lag}\right)^2-\left(\R_0^{\eul}\right)^2}{\left(\R_0^{\eul}\right)^2} = \frac{\rho(P^{\T}\mathcal{M})-\rho(\mathcal{M})}{\rho(\mathcal{M})}\leq \frac{1}{4}\frac{\delta_v}{ (2m_v + \delta_v)} \, .
\end{split}
\end{equation}

\end{enumerate}
\end{proposition}

From (\ref{Bound0}) and using that $\R_0^{\lag}>\R_0^{\eul}$, we get the following bound for the relative difference $\mathcal{R}_0^{\lag}$ with respect to $\mathcal{R}_0^{\eul}$:  

$$\frac{\mathcal{R}_0^{\lag}- \mathcal{R}_0^{\eul}}{\mathcal{R}_0^{\eul}} \leq  \frac{1}{4}\frac{\R_0^{\eul}}{\R_0^{\lag}+\R_0^{\eul}} \leq \frac{1}{8} \, .$$
In consequence, the percentage difference between the basic reproduction numbers for systems (\ref{emodel1}) and (\ref{emodel2}) is at most $12.5\% $ under homogeneous infection if $0<p_{12},p_{21}<1/2$. In addition, the larger $|p_{12}-p_{21}|$ is, the larger the difference between $\mathcal{R}_0^{\eul}$ and $\mathcal{R}_0^{\lag}$ is as well, as Figure \ref{fr0lagr0eul} shows. 

\begin{figure}[H]
\centering
\includegraphics[scale=0.5]{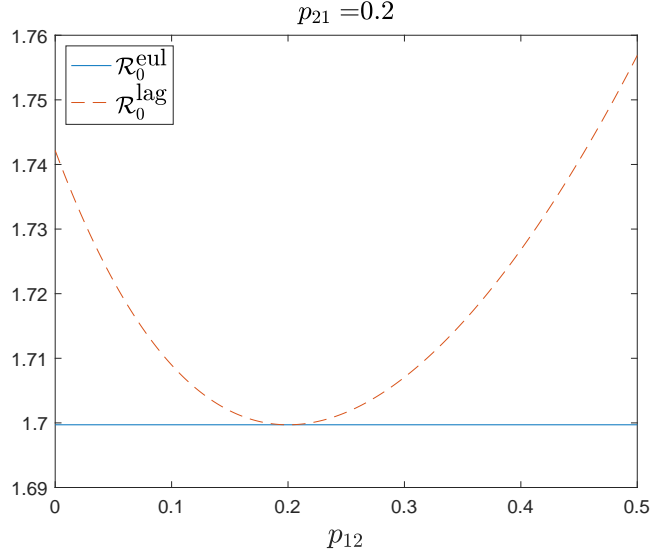}
\caption{Basic reproduction number for system (\ref{emodel1}) in the two-patch case as a function of $ p_{12}$ with $a=0.3, b=0.1, c = 0.214, \delta = r = 1/150, \delta_v = 0.1, m := \left(N_{v,i}^{\lag}\right)^*/\hat{N_i}^*=\left(N_{v,i}^{\eul}\right)^*/\left(N_i^{\eul}\right)^* = 1, \beta = ab, \beta_v = acm $ \cite{ruktanonchai2016identifying}, and units as in Table \ref{param1}. In this case, the largest percentage  difference has value $100 \left(\mathcal{R}_0^{\lag}- \mathcal{R}_0^{\eul}\right)/\mathcal{R}_0^{\eul} = 3.36\%$ and is attained at $p_{12} = 1/2$.}
\label{fr0lagr0eul}
\end{figure}

In conclusion, under homogeneous conditions for both patches, the introduction of infectious individuals creates more secondary infections according to the Lagrangian dynamics for any matrix $P$. By contrast, if we suppose $\delta_{v,1} > \delta_{v,2}$, it is then possible that $\R_0^{\lag} < \R_0^{\eul}$ when $p_{12} > p_{21}$ (see Figure \ref{figdeltas}). 

\begin{figure}[H]
\centering
\includegraphics[scale=0.5]{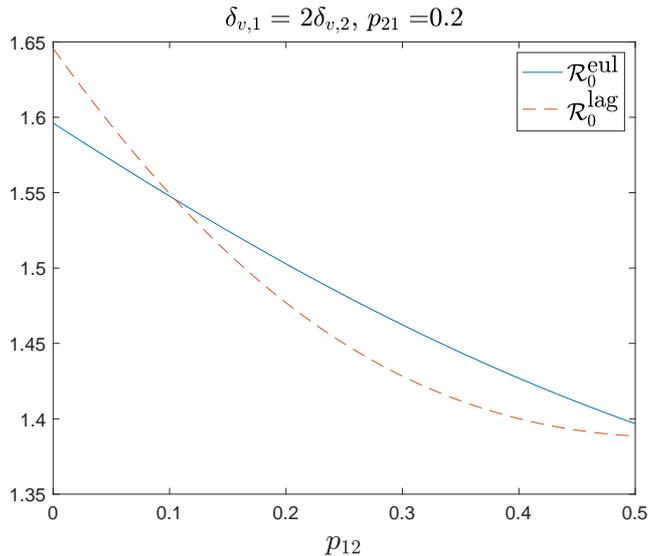}
\caption{ Comparison reproduction numbers for $\delta_{v,1} > \delta_{v,2}$. The used parameters are $a=0.3, b=0.1, c=0.214, \delta_1 = \delta_2 = 1/150, \delta_{v,1} =0.1, \beta_1 = \beta_2 = ab, \beta_{v,1} = \beta_{v,2} = ac,m_v =0.02$ \cite{ruktanonchai2016identifying}, and the corresponding units are as in Table \ref{param1}. 
}
\label{figdeltas}
\end{figure}

%%%%%%%%%%%%%%%%%%%%%%%%%%%%%%%%%%%%%%%%%%%
\section{Examples using data} \label{SdataP}
%%%%%%%%%%%%%%%%%%%%%%%%%%%%%%%%%%%%%%%%%%%

We now turn to applying our analytical results and definition of consistency to empirical data.  There is an abundance of empirical data on mobility and connectivity between locations \cite{bengtsson2011improved, lessler2015seven, Wesolowski2012MalariaMove, wesolowski2015impact}, and these data are being increasingly incorporated into mathematical and computational models of vector-borne disease dynamics \cite{iggidr2017vector, ruktanonchai2016identifying}.  Many factors are involved for deciding whether to use a Lagrangian or Eulerian modeling framework, including the spatial scale involved, mathematical tractability, and type of data available.  Here we consider two empirical data sets (Sections \ref{Snamibia}, \ref{Sbrazil}) and two hypothetical data sets (Sections \ref{Siquitos}, \ref{Swestnile}) on host movement in the context of vector-borne disease on spatial scales ranging from within village \cite{vazquez2009usefulness} to country-wide \cite{ruktanonchai2016identifying}, with data sources including long-term GPS trackers \cite{vazquez2009usefulness} and mobile phones \cite{ruktanonchai2016identifying}.  Studies in \cite{iggidr2017vector,ruktanonchai2016identifying} incorporate the corresponding data into a Lagrangian modeling framework. Here we examine whether the data are consistent with an Eulerian framework in the sense of Definition \ref{def:consistency}, and discuss possible implications of using an Eulerian approach for these specific settings in terms of the basic reproduction number.

%%%%%%%%%%%%%%%%%%%%%%%%%%%%%%%%%%%%%%%%%%%
\subsection{Malaria in Namibia} \label{Snamibia}
%%%%%%%%%%%%%%%%%%%%%%%%%%%%%%%%%%%%%%%%%%%

Ruktanonchai et al. \cite{ruktanonchai2016identifying} use mobile phone records to examine movement between health districts in Namibia, in the context of malaria control efforts.  Specifically, the authors identified mobility sources / sinks from mobile phone records, together with transmission hot spots from malaria parasite maps.  The mobility data and local transmission parameters are combined in a Lagrangian framework for vector-host dynamics \cite{ruktanonchai2016identifying} . 

Anonymized mobile phone records were collected over a year from 1.2 million phones, corresponding to approximately 85\% of the adult population in Namibia.  Call and SMS data were used to identify locations at the health district level.  Home health districts and location changes were estimated, and aggregated to produce a Lagrangian mixing matrix $P$ \cite{ruktanonchai2016identifying}.  
For most of the locations in these data we have that the quantities $F_i^{(\textrm{in})}:=\sum_{j\neq i}p_{ij}$ and $F_i^{(\textrm{out})} := \sum_{j\neq i}p_{ji}$ are small. \\

Using the data considered in \cite{ruktanonchai2016identifying}, we take the ten health districts for which the quantity $\left(F_i^{(\textrm{in})}\right)^2+\left(F_i^{(\textrm{out})}\right)^2$ is the largest and define the $10\times 10$ mixing matrix $P$ for these locales. Parameter values used in this example are as in \cite{ruktanonchai2016identifying}:
$a=0.3 \hspace{.2cm} \text{Hosts} \times \text{Days}^{-1}$, $b=0.1 \hspace{.2cm} \text{Vectors}^{-1}$, $c = 0.214 \hspace{.2cm} \text{Hosts}^{-1}$, $\delta = r = 1/150 \hspace{.2cm} \text{Days}^{-1}$, $\delta_v = 0.1 \hspace{.2cm} \text{Days}^{-1}$, $m := \left(N_{v,i}^{\lag}\right)^*/\hat{N_i}^*=\left(N_{v,i}^{\eul}\right)^*/\left(N_i^{\eul}\right)^* = 1$, $\beta = ab \hspace{.2cm} \text{Days}^{-1}$, $\beta_v = acm \hspace{.2cm} \text{Days}^{-1}$. 
In this case, systems (\ref{emodel1}) and (\ref{emodel2}) are inconsistent with small relative error $\norm{L - \left(D_\delta^{\lag}P^{-1} - D_\delta^{\eul} \right)}_F/\norm{D_\delta^{\lag}P^{-1} - D_\delta^{\eul}}_F=0.018$, where $L$ is the solution of (\ref{eqn:E_min}). In addition, the Lagrangian basic reproduction number $\mathcal{R}_0^{\lag}=1.7049$ is slightly larger than the Eulerian basic reproduction number $\mathcal{R}_0^{\eul} = 1.6997$. In conclusion, the basic reproduction numbers for both systems are similar, even though  the systems are not consistent (with small relative error).

%%%%%%%%%%%%%%%%%%%%%%%%%%%%%%%%%%%%%%%%%%%
\subsection{Dengue in Brazil} \label{Sbrazil}
%%%%%%%%%%%%%%%%%%%%%%%%%%%%%%%%%%%%%%%%%%%

Iggidr et al \cite{iggidr2017vector} use a Lagrangian framework to model dengue in eight locations forming the metropolitan area of Rio de Janeiro. The host density in each location was determined from the national census, and a 20\% host-vector ratio was assumed in their simulations.  The mixing matrix was estimated
from data provided by the transport authority of Rio de Janerio  
(see Appendix B of \cite{iggidr2017vector}).

Consider the parameter values $\delta_i = 1/10.5 \hspace{.2cm} \text{Days}^{-1}$, $\delta_{v,i} = 1/5.5 \hspace{.2cm} \text{Days}^{-1}$, $ m :=\left(N_{v,i}^{\eul}\right)^*/\left(N_i^{\eul}\right)^* = 1$, $\beta_i = 0.3750 \hspace{.2cm} \text{Days}^{-1}$, $ \beta_{v,i} = 0.3750 \hspace{.2cm} \text{Days}^{-1} $ used in \cite{lee2015role}. In this case, systems (\ref{emodel1}) and (\ref{emodel2}) are inconsistent albeit with relative error $\norm{ L - \left(D_\delta^{\lag}P^{-1} - D_\delta^{\eul} \right)}_F/\norm{ D_\delta^{\lag}P^{-1} - D_\delta^{\eul}}_F=0.0277$, where $L$ is the solution of (\ref{eqn:E_min}). Moreover, the Lagrangian basic reproduction number $\mathcal{R}_0^{\lag}=2.8586$ is nearly equal to the Eulerian basic reproduction number $\mathcal{R}_0^{\eul} = 2.8498$. Thus despite the systems being inconsistent, similar reproduction numbers are obtained for the Eulerian and Lagrangian modeling frameworks, suggesting flexibility in using either framework in terms of the domain basic reproduction number.  

%%%%%%%%%%%%%%%%%%%%%%%%%%%%%%%%%%%%%%%%%%%
\subsection{Dengue in Iquitos, hypothetical mixing matrix} \label{Siquitos}
%%%%%%%%%%%%%%%%%%%%%%%%%%%%%%%%%%%%%%%%%%%

In \cite{vazquez2009usefulness}, neck trap GPS devices 
were used to register movements of a carpenter and mototaxi driver in Iquitos, Peru for approximately two weeks. The spatial scale of movement here is small, suggesting a Lagrangian framework. The second column of Table \ref{tablecarpenter} (`Time') corresponds to the time spent by the carpenter in $4$ houses (P1, P2, P4, P5) during 14 days according to the GPS data. The third column of Table \ref{tablecarpenter} (`Proportions') corresponds to the ratios between the times in the second column and the total hours in 14 days.  

\begin{table}[H]
\centering
\begin{tabular}{|c|c|c|}
\hline
House & Time (hours) & Proportion \\
\hline
P1 & 12.1 & 0.036\\
\hline
P2 & 3 & 0.008\\
\hline
P4 & 5.4& 0.016\\
\hline
P5 & 18 & 0.053\\
\hline
\end{tabular}
\caption{Time and proportion of time spent in the most frequented four visited houses (other than home) by a carpenter during two weeks in Iquitos \cite{vazquez2009usefulness}.}
\label{tablecarpenter}
\end{table}

The mixing matrix $P$ is a hypothetical arrangement based on the proportions in Table \ref{tablecarpenter}. 

\begin{equation}\label{eqn:PIquitos}
 P= \begin{pmatrix}
 0.887 &   0.036 &   0.036  &  0.036  &  0.036 \\
 0.036 &   0.887  &  0.008  &  0.008  &  0.008 \\
 0.008 &   0.008  &  0.887  &  0.016  &  0.016 \\
 0.016 &   0.016  &  0.016 &   0.887  &  0.053 \\
 0.053 &   0.053 &   0.053  &  0.053  &  0.887
\end{pmatrix},
\end{equation}

$$L = \begin{pmatrix}
    0.0241  & -0.0077 &  -0.0077 &  -0.0077 &  -0.0077 \\
   -0.0082  &  0.0236 &  -0.0014 &  -0.0014 &  -0.0014 \\
   -0.0015  & -0.0015 &   0.0235 &  -0.0034 &  -0.0034 \\
   -0.0028  & -0.0028 &  -0.0028 &   0.0241 &  -0.0121 \\
   -0.0116  & -0.0116 &  -0.0116 &  -0.0116 &   0.0246
\end{pmatrix}. $$

In this example we use the parameters $\delta_i = 1/10.5 \hspace{.2cm} \text{Days}^{-1}$, $\delta_{v,i} = 1/5.5 \hspace{.2cm} \text{Days}^{-1}$, $m := \left(N_{v,i}^{\lag}\right)^*/\hat{N_i}^*=\left(N_{v,i}^{\eul}\right)^*/\left(N_i^{\eul}\right)^* = 1$, $\beta_i = 0.3750 \hspace{.2cm} \text{Days}^{-1}$, $\beta_{v,i} = 0.3750 \hspace{.2cm} \text{Days}^{-1}$ given in \cite{lee2015role}. Here we obtain that systems (\ref{emodel1}) and (\ref{emodel2}) are consistent, i.e, $\norm{L - \left(D_\delta^{\lag}P^{-1} - D_\delta^{\eul} \right) }_F/\norm{D_\delta^{\lag}P^{-1} - D_\delta^{\eul} }_F=0$. Furthermore, the Lagrangian basic reproduction number $\mathcal{R}_0^{\lag}=2.8536$ is approximately the Eulerian basic reproduction number $\mathcal{R}_0^{\eul} = 2.8498$.

%%%%%%%%%%%%%%%%%%%%%%%%%%%%%%%%%%%%%%%%%%%
\subsection{Migratory hosts, hypothetical mixing matrix}\label{Swestnile}
%%%%%%%%%%%%%%%%%%%%%%%%%%%%%%%%%%%%%%%%%%%

We consider the hypothetical $3 \times 3$ mixing matrix (\ref{matrixexample}) from Section \ref{sectioninconsistent}, where the hosts in all the three patches spend most of their time in patch $1$. A mixing matrix such as (\ref{matrixexample}) could correspond to the movement of migratory hosts that do not have a sense of home.  In this situation an Eulerian modeling framework is natural. Specifically, consider systems (\ref{emodel1}) and (\ref{emodel2}) with $P$ and $L$ from (\ref{matrixexample}) and patch parameters corresponding to West Nile Virus in migratory birds, as used in \cite{bergsman2016mathematical}. These parameters are $\delta_i = 0.2222 \hspace{.2cm} \text{Days}^{-1}$, $\delta_{v,i} = 0.0666 \hspace{.2cm} \text{Days}^{-1}$, $m := \left(N_{v,i}^{\lag}\right)^*/\hat{N_i}^*=\left(N_{v,i}^{\eul}\right)^*/\left(N_i^{\eul}\right)^* = 1$, $\beta_i = 0.2479 \hspace{.2cm} \text{Days}^{-1}$, $\beta_{v,i} = 0.2479 \hspace{.2cm} \text{Days}^{-1}$. For this example, systems (\ref{emodel1}) and (\ref{emodel2}) with the largest possible relative error, ($L=0$ the solution to the minimization problem \eqref{eqn:E_min}). The Lagrangian basic reproduction number $\mathcal{R}_0^{\lag}= 1.5796 $ is approximately $21 \%$ larger than the Eulerian basic reproduction number $\mathcal{R}_0^{\eul} = 1.3135$. In conclusion, in this example the systems are inconsistent and the connected movement network for the Lagrangian system is not reflected in the disconnected movement network for the Eulerian system. In addition, the difference in the values of the basic reproduction numbers from both systems may be significant. 

Table \ref{dataP} summarizes the comparison between the estimated reproduction numbers in Sections \ref{Snamibia}, \ref{Sbrazil}, \ref{Siquitos} and \ref{Swestnile}, showing that in the inconsistent cases the Lagrangian model gives a larger basic reproduction number, and the differences $\R_0^{\lag} - \R_0^{\eul}$  increases as the error increases for these examples.  For each of these examples, the basic reproduction numbers for the Eulerian and Lagrangian frameworks are similar to one another.  By contrast, in Section \ref{Swestnile}, we have an example where host movement in Eulerian system that we obtain from (\ref{eqn:E_min}) is totally disconnected (i.e., $L=0$), and the difference between the basic reproduction numbers may be significant.

\begin{table}[H]
\centering
\begin{tabular}{llll}
\hline\noalign{\smallskip}
Example &   $E_r$ & $\mathcal{R}_0^{\lag}$ & $\mathcal{R}_0^{\eul}$ \\
\noalign{\smallskip}\hline\noalign{\smallskip}

Malaria data in Namibia in \cite{ruktanonchai2016identifying} & 0.018 & 1.7049 & 1.6997\\
Brazil transportation data in \cite{iggidr2017vector} & 0.0277 & 2.8586 & 2.8498 \\
Dengue in Iquitos, hypothetical $P$ & 0 & 2.8536 & 2.8498\\
Migratory hosts, hypothetical $P$  & 1 & 1.5796 & 1.3135\\
\noalign{\smallskip}\hline
\end{tabular}
\caption{ Relative error $E_r = \norm{ L - \left(D_\delta^{\lag}P^{-1}-D_\delta^{\eul}\right) }_F/\norm{D_\delta^{\lag}P^{-1}-D_\delta^{\eul} }_F$ and values of the basic reproduction numbers $\mathcal{R}_0^{\lag}$ and $\mathcal{R}_0^{\eul}$ in the examples of Sections \ref{Snamibia}, \ref{Sbrazil}, \ref{Siquitos} and \ref{Swestnile}.}
\label{dataP}
\end{table}
%

%%%%%%%%%%%%%%%%%%%%%%%%%%%%%%%%%%%%%%%%%%%
\section{Discussion} 
\label{sec:conclusions}
%%%%%%%%%%%%%%%%%%%%%%%%%%%%%%%%%%%%%%%%%%%

Lagrangian and Eulerian approaches are important modeling tools for studying the effects of heterogeneity and movement in disease dynamics \cite{cosner2015models}. 
We have presented an approach for relating the Eulerian and Lagrangian systems through a fundamental matrix by matching the time that infectious individuals reside in other patches. We define the Eulerian and Lagrangian systems to be consistent when the fundamental matrices match, in the sense that the minimum value of the optimization problem \eqref{eqn:E_min} is zero. 

As the star graph example in Section \ref{sec:consistency} and mixing matrix in Section \ref{Swestnile} show, both consistency and inconsistency between the two frameworks is possible.  While we do not have a complete characterization of when the Eulerian and Lagrangian frameworks are consistent, Proposition \ref{prop:sufficient} gives a sufficient condition.  Specifically, Proposition \ref{prop:sufficient} gives intervals $[p_*,p^*]$ such that if all the off-diagonal elements of the mixing matrix are in $[p_*,p^*]$, then the systems are consistent.  The upper bound in the sufficiency criterion can be interpreted in terms of Lagrangian models being suitable for situations where individuals commute from a distinguished home location.  This setting often corresponds to individuals spending the majority of their time in the home location, meaning that the off-diagonal entries of the mixing matrix are small \cite{cosner2015models}.  Inconsistent examples with large off-diagonal elements of the mixing matrix such as in Section \ref{Swestnile} thus conflict with the sense of home that Lagrangian models try to capture. Inconsistency is also possible when the off-diagonal entries are small, as for the mixing matrix in \eqref{eqn:examplesufficient}. Identifying additional necessary criteria for consistency is an area for future work.

In \cite{iggidr2017vector} it is discussed how to go from an Eulerian to a Lagrangian framework where the movement rates are relatively larger than the removal rates. The Eulerian framework considered in \cite{iggidr2017vector} is different from the Eulerian framework considered in this paper.  
Specifically, \cite{iggidr2017vector} use an Eulerian framework with $n^2$ variables corresponding to residents of patch $i$ that are currently located in patch $j$.  For example, $S^h_{ij}(t)$ represents the number of susceptible hosts whose home is patch $i$ and are in patch $j$ at time $t$. The movement rates corresponding to $S^h_{ij}$ are $m_{kj}^{i}$ for $j\neq k$ (movement from $j$ to $k$), from which we can define a graph Laplacian $L_i$ for each home patch $i$. A mixing matrix $P$ can be obtained  from  $L_1,\ldots, L_n$ under the assumption that the movement rates are large compared to the removal rates. This is different than the Eulerian framework we consider, where the movement rates are captured by a single graph Laplacian $L^{\xx}$ for every state $\xx\in\{\Su,\In,\re\}$. In addition, the consistency definition here requires the off-diagonal entries of $D_\delta^{\lag}P^{-1}-D_\delta^{\eul}$ to be non-positive. Therefore, for given $P$ and $D_\delta^{\eul}$, consistency depends only upon the sign of the off-diagonal elements of $P^{-1}$ and does not depend on the removal rates in $D_\delta^{\eul}$. In consequence, in contrast to the time scales assumption in \cite{iggidr2017vector}, the concept of consistency that we present does not depend on the relative timescales of movement to removal.
 
The functional implications of using a Lagrangian versus an Eulerian approach are important to consider.  We find that the domain $\R_0$ values are similar under various scenarios when the two approaches are consistent. In the homogeneous consistent case (Section \ref{sec:comparisonhomo}), we obtain explicit bounds (Proposition \ref{prep:homogeneous}) for the difference in $\R_0$ between the Eulerian and Lagrangian frameworks. Furthermore, the behavior of the Lagrangian basic reproduction number in Proposition \ref{prep:homogeneous} of Section \ref{sec:comparisonhomo} is consistent with studies such as \cite{lee2015role} (in the sense of attaining a minimum value when $p_{12} = p_{21}$, see Fig. 4 of \cite{lee2015role}). 

Although there is inconsistency in the examples of Sections \ref{sectioninconsistent}, \ref{Snamibia}, \ref{Swestnile}, we observe that the obtained values of basic reproduction number are still alike.  
This suggests using \eqref{eqn:E_min} to relate mixing matrices such as those given in \cite{iggidr2017vector,ruktanonchai2016identifying} to Eulerian systems, and then studying the reproduction number of the resulting Eulerian systems using techniques such as those in \cite{jacobsen2018generalized,tien2015disease}, can be an effective approach.  The graph Laplacian matrix $L$ obtained from the optimization problem (\ref{eqn:E_min}) (as in Sections \ref{Snamibia} and \ref{Sbrazil}) allows series expansions for the basic reproduction number, and derivation of important quantities such as the absorption inverse $L^d$ \cite{jacobsen2018generalized} for analyzing the mobility network. For example, $L^d$ captures the effect of absorption (that are the removal rates in this case) on the movement network and also has applications in community detection and node centrality \cite{benzi2019graphs, jacobsen2018generalized}, leading to extensions of the analysis of the mixing matrix $P$. 

In contrast to $\R_0$, we observe significant differences in outbreak size between the Eulerian and Lagrangian approaches.  Indeed, differences in outbreak size can be large not only when the systems are inconsistent (e.g. Figure \ref{figfinalsize}), but for the consistent case as well.  Therefore, care must be taken when going from one approach to the other.  It would be useful to have a bound for the differences in outbreak size between the two approaches. Analytical results quantifying how different the final sizes are is an area for future work.

%%%%%%%%%%%%%%%%%%%%%%%%%%%%%%%%%%%%%%%%%%%
\section*{Declaration of Competing Interest}
%%%%%%%%%%%%%%%%%%%%%%%%%%%%%%%%%%%%%%%%%%%
None.

%%%%%%%%%%%%%%%%%%%%%%%%%%%%%%%%%%%%%%%%%%%
\section*{Acknowledgements}
%%%%%%%%%%%%%%%%%%%%%%%%%%%%%%%%%%%%%%%%%%%
The authors are grateful to Chris Cosner for feedback on an early version of this manuscript.
This work was supported by the National Science Foundation (DMS 1814737, DMS 1440386), and the Fullbright International Fellow Program.  

%%%%%%%%%%%%%%%%%%%%%%%%%%%%%%%%%%%%%%%%%%%
\section{Appendix } \label{Sappendix}
%%%%%%%%%%%%%%%%%%%%%%%%%%%%%%%%%%%%%%%%%%%

%%%%%%%%%%%%%%%%%%%%%%%%%%%%%%%%%%%%%%%%%%%
\subsection{Basic reproduction number} \label{appendixr0}
%%%%%%%%%%%%%%%%%%%%%%%%%%%%%%%%%%%%%%%%%%%

In this section of the Appendix we use the next generation matrix approach \cite{van2002reproduction} to derive the expressions for the basic reproduction numbers (\ref{er01}) and (\ref{er02}) of systems (\ref{emodel1}) and (\ref{emodel2}) respectively. 

We first compute the next generation matrix $\left(F^{\lag}\right)\left(V^{\lag}\right)^{-1}$ of the Lagrangian system. Let us consider the equations 
\begin{align} \label{eqn:system1Inf}
\dot{I_i} & = \sum_{j=1}^n \beta_j p_{ji}\frac{S_i}{N_i}I_{v,j} - \left(\gamma_i^{\lag} + \mu_i^{\lag}\right)I_i \notag \\
\dot{I}_{v,i} & = \beta_{v,i} \frac{\sum_{j=1}^n p_{ij}I_j}{\sum_{j=1}^np_{ij}N_j} S_{v,i} + \sum_{j=1}^n m_{ij}^v I_{v,j} - \sum_{j=1}^n m_{ji}^v I_{v,i} - \mu_{v,i} I_{v,i}
\end{align}
corresponding to the infectious compartments of system (\ref{emodel1}). From (\ref{eqn:system1Inf}) we define the function $\mathcal{F}^{\lag}: \mathbb{R}^{2n} \rightarrow \mathbb{R}^{2n}$ by  $$\mathcal{F}_i^{\lag}(I_1,\ldots,I_n,I_{v,1},\ldots,I_{v,n}):= \sum_{j=1}^n \beta_j p_{ji}\frac{S_i}{N_i}I_{v,j}\,,$$ for $i=1,\ldots,n$, and $$\mathcal{F}_i^{\lag}(I_1,\ldots,I_n,I_{v,1},\ldots,I_{v,n}) := \beta_{v,i} \frac{\sum_{j=1}^n p_{ij}I_j}{\sum_{j=1}^np_{ij}N_j} S_{v,i} \,, $$ for $i=n+1, \ldots,2n$. The DFE of the Lagrangian systems is determined by $\left(S_i^{\lag}\right)^* = \left(N_i
^{\lag}\right)^*$ and $\left(S_{v,i}^{\lag}\right)^* = \left(N_{v,i}^{\lag}\right)^*$ defined by (\ref{eqn:DFElagrangian}). We then define the Jacobian matrix $$F^{\lag} := \left. \partial \mathcal{F}^{\lag} / \partial (I_1,\ldots, I_n, I_{v,1}, \ldots, I_{v,n})\right|_{DFE}\,.$$ We have that $\left.\left(\partial \mathcal{F}_i^{\lag} / \partial I_{v,j}\right)\right|_{DFE} = \beta_j p_{ji}$, so the upper-right block of $F^{\lag}$ is $P^{\T} D_{\beta}$, where $D_{\beta}:= \diag\{\beta_i\}$. We also have that $\left. \left(\partial \mathcal{F}_{n+i}^{\lag} / \partial I_j\right)\right|_{DFE} = \beta_{v,i} \left(N_{v,i}^*/\hat{N_i}^*\right) p_{ij} $, so the lower-left block of $F^{\lag}$ is $D_{\beta_v}^{\lag} P$, where $\hat{N_i}^*:= \sum_{j=1}^n p_{ij} \left(N_j^{\lag}\right)^*$ and $D_{\beta_v}^{\lag} := \diag\left\{\beta_{v,i} N_{v,i}^*/\hat{N_i}^* \right\}$. Therefore, we have $$F^{\lag} = \begin{pmatrix}
0 & P^{\T} D_{\beta} \\
D_{\beta_v}^{\lag}P & 0
\end{pmatrix}. $$

Similarly, we define the function $\mathcal{V}^{\lag}: \mathbb{R}^{2n} \rightarrow \mathbb{R}^{2n}$ by $$\mathcal{V}_i^{\lag} := \left(\gamma_i^{\lag} + \mu_i^{\lag}\right)I_i \,,$$ for $i = 1,\ldots,n$, and
$$\mathcal{V}_i^{\lag} :=  - \sum_{j=1}^n m_{ij}^v I_{v,j} + \sum_{j=1}^n m_{ji}^v I_{v,i} + \mu_{v,i} I_{v,i}\,, $$ for $i = n+1,\ldots, 2n$. We also define the Jacobian matrix $$V^{\lag} := \left.\partial \mathcal{V}^{\lag}/ \partial (I_1.\ldots,I_n, I_{v,1},\ldots, I_{v,n})\right|_{DFE}\,.$$ If $L_v$ is the graph Laplacian of the vector movement (with adjacency matrix $M^{v} = \left(m_{ij}^v\right)_{i,j\leq n}$, see (\ref{eqn:laplacian})) and $D_\delta^{\lag}:= \diag\left\{\gamma_i^{\lag}+\mu_i^{\lag}\right\}$, then the upper-left block of $V^{\lag}$ is $D_\delta^{\lag}$ and the lower-right block of $V^{\lag}$ is $G_v = L_v + D_{\delta_v}$. Therefore, $$ V^{\lag} = \begin{pmatrix}
D_\delta^{\lag} & 0 \\
0 & G_v
\end{pmatrix}. $$
Consequently,  $$ \left(F^{\lag}\right)\left(V^{\lag}\right)^{-1} = \begin{pmatrix}
0 & P^{\T}D_{\beta}G_v^{-1} \\
D_{\beta_v}^{\lag}P \left(D_\delta^{\lag}\right)^{-1} & 0
\end{pmatrix} $$
and
\begin{equation*}
(\mathcal{R}_0^{\lag})^2 = \rho\left(P^{\T}D_{\beta}G_v^{-1}D_{\beta_v}^{\lag}P \left(D_\delta^{\lag}\right)^{-1}\right) \,.
\end{equation*}

We now compute the next generation matrix $\left(F^{\eul}\right)\left(V^{\eul}\right)^{-1}$ of system (\ref{emodel2}). The equations of the infectious compartments of system (\ref{emodel2}) are 

\begin{align}\label{eqn:system2Inf}
\dot{I_i} & =  \beta_i \frac{S_i}{N_i}I_{v,i} + \sum_{j=1}^n m_{ij}^{\In} I_{j} - \sum_{j=1}^n m_{ji}^{\In} I_{i}  - \left(\gamma_i^{\eul} + \mu_i^{\eul}\right)I_i \notag \\
\dot{I_{v,i}} & = \beta_{v,i} \frac{I_i}{N_i} S_{v,i} + \sum_{j=1}^n m_{ij}^v I_{v,j} - \sum_{j=1}^n m_{ji}^v I_{v,i} - \mu_{v,i} I_{v,i}\, .
\end{align}

From the equations in (\ref{eqn:system2Inf}) we define the function $$\mathcal{F}^{\eul}(I_1,\ldots, I_n, I_{v,1},\ldots,I_{v,n}):=  \beta_i \frac{S_i}{N_i}I_{v,i}\,,$$ for $i=1,\ldots, n$, and $$\mathcal{F}^{\eul}(I_1,\ldots, I_n, I_{v,1},\ldots,I_{v,n}) = \beta_{v,i} \frac{I_i}{N_i} S_{v,i}\,, $$ for $i=n+1,\ldots, 2n$. 

Using the DFE (\ref{eqn:DFEeulerian}) of system (\ref{emodel2}), we have that $\left.\left(\partial \mathcal{F}_i^{\eul} / \partial I_{v,j}\right)\right|_{DFE} = \beta_i $ and 
$\left. \left(\partial \mathcal{F}_{n+i}^{\eul} / \partial \mathcal{F}_{j}^{\eul}\right)\right|_{DFE} = \beta_{v,i} N_{v,i}^*/\left(N_i^{\eul}\right)^* \,.$ Therefore, if $$F
^{\eul}:= \left. \partial \mathcal{F}^{\eul} /\partial (I_1,\ldots, I_n, I_{v,1},\ldots, I_{v,n})\right|_{DFE}\,,$$ we  then have $$ F^{\eul} = \begin{pmatrix}
0 & D_\beta \\
D_{\beta_v}^{\eul} & 0
\end{pmatrix} \, , $$
where $D_{\beta}:= \diag\{\beta_i\}$ and  $D_{\beta_v}^{\eul} := \diag\left\{\beta_{v,i} N_{v,i}^*/\left(N_i^{\eul}\right)^* \right\}$. Similarly, we define the function $$\mathcal{V}^{\eul}(I_1,\ldots, I_n, I_{v,1},\ldots, I_{v,n}) := - \sum_{j=1}^n m_{ij}^{\In} I_{j} + \sum_{j=1}^n m_{ji}^{\In} I_{i}  + \left(\gamma_i^{\eul} + \mu_i^{\eul}\right)I_i \,, $$ for $i = 1,\ldots,n$, and
$$\mathcal{V}^{\eul}(I_1,\ldots, I_n, I_{v,1},\ldots, I_{v,n}) := - \sum_{j=1}^n m_{ij}^v I_{v,j} + \sum_{j=1}^n m_{ji}^v I_{v,i} + \mu_{v,i} I_{v,i}\,,$$ for $i = n+1,\ldots, 2n$. Therefore, if $$V^{\eul}:= \left.\partial \mathcal{V}^{\eul} / \partial(I_1,\ldots,I_n,I_{v,1}\ldots, I_{v,n})\right|_{DFE}\,,$$ we then have 
\begin{equation}\label{veul}
V^{\eul} = \begin{pmatrix}
G & 0\\
0 & G_v
\end{pmatrix}\, ,
\end{equation}

\noindent where $G := L^{\In} + D_{\delta}^{\eul}$,  $L^{\In}$ is the graph Laplacian of the host movement [with adjacency matrix $M^{\In} = (m_{ij}^{\In})_{i,j\leq n}$, see (\ref{eqn:laplacian})], $D_\delta^{\eul}:= \diag\{\gamma_i^{\eul}+\mu_i^{\eul}\}$, and $G_v = L_v + D_{\delta_v}$. In consequence,
$$ \left(F^{\eul}\right)\left(V^{\eul}\right)^{-1} = \begin{pmatrix}
0 & D_\beta G_v^{-1}\\
D_{\beta_v}^{\eul}G^{-1} &0
\end{pmatrix}$$

\noindent and \begin{equation*}
\left(\mathcal{R}_0^{\eul}\right)^2=\rho\left(D_\beta G_v^{-1}D_{\beta_v}^{\eul}G^{-1}\right).
\end{equation*}

%%%%%%%%%%%%%%%%%%%%%%%%%%%%%%%%%%%%%%%%%%%
\subsection{Comparison of basic reproduction numbers}\label{appendixhomogeneous}
%%%%%%%%%%%%%%%%%%%%%%%%%%%%%%%%%%%%%%%%%%%

In this section we prove Proposition \ref{prep:homogeneous} of Section \ref{sec:comparisonhomo}. Let $\beta_{v,1} = \beta_{v,2}$, $N_v^* := N_{v,1}^*=N_{v,2}^*$, $\left(N^{\eul}\right)^* := \left(N_1^{\eul}\right)^* = \left(N_2^{\eul}\right)^*$, and $\beta_v =  \beta_{v,1} N_{v,1}^*/\left(N_1^{\eul}\right)^* = \beta_{v,2} N_{v,2}^*/\left(N_2^{\eul}\right)^*$. Define $\left(N_1^{\lag}\right)^*$ and $\left(N_2^{\lag}\right)^*$  such that $\left(N_1^{\eul}\right)^*= p_{11}\left(N_1^{\lag}\right)^* + p_{12}\left(N_2^{\lag}\right)^*$ and $\left(N_2^{\eul}\right)^* = p_{21}\left(N_1^{\lag}\right)^* + p_{22}\left(N_2^{\lag}\right)^*$. Hence  A\ref{a:DFEs} holds, and $\delta = \delta_1^{\eul} = \delta_2^{\eul}$, so we also get $\delta = \delta_1^{\lag} = \delta_2^{\lag}$ by A\ref{a:abs}. Let $D_\beta = \beta I$,  $D_{\beta_v} = \beta_v 	I$, $D_\delta^{\eul} = D_\delta^{\lag} = D_\delta := \delta I$, $D_{\delta_v} =\delta_v I$, $L_v = m_v \begin{pmatrix}
1 & -1\\ -1 & 1\end{pmatrix}, P = \begin{pmatrix} 1-p_{21} & p_{12} \\ p_{21} & 1-p_{12} \end{pmatrix}$. 

We fix all the parameters except $p_{12}$ (we obtain analogous and symmetric results if we fix $p_{21}$). From (\ref{er01}) and (\ref{er02}), we obtain that $\left(\mathcal{R}_{0}^{\lag}\right)^2 = \rho\left(P^{\T}D_{\beta}G_v^{-1}D_{\beta_v}P \left(D_\delta\right)^{-1}\right) $  and  $ \left(\mathcal{R}_{0}^{\eul}\right)^2=\rho\left(D_\beta G_v^{-1}D_{\beta_v}P(D_\delta \right)^{-1}) $.
Therefore, if we define $\mathcal{M}:=D_{\beta}G_v^{-1}D_{\beta_v}P \left(D_\delta^{\lag}\right)^{-1}$, we get

\begin{equation}
\left(\mathcal{R}_{0}^{\lag}\right)^2=\rho(P^{\T}\mathcal{M}) \quad \textrm{and} \quad \left(\mathcal{R}_{0}^{\eul}\right)^2 = \rho(\mathcal{M})\,.
\end{equation} 
The following is the statement of the proposition.  
\setcounter{proposition}{1}

\begin{proposition}
Assume A\ref{a:M}--A\ref{a:Pinvertible} and suppose that $D_\beta = \beta I$, $D_{\beta_v} = \beta_v 	I$, $D_\delta^{\eul} = D_\delta^{\lag} = D_\delta := \delta I$, $D_{\delta_v} =\delta_v I$ and $m_{12}^v = m_{21}^v = m_v$. Define $\mathcal{M}:=D_{\beta}G_v^{-1}D_{\beta_v}P \left(D_\delta^{\lag}\right)^{-1}$ and fix $p_{21}$ in the interval $(0,1/2)$. Then $\rho(P^{\T}M)$ is a function of $p_{12}$ where $0<p_{12}<1/2$  and we have that:
\begin{enumerate}[a)]
    \item $\rho(\mathcal{M}) = (\beta \beta_v)/(\delta \delta_v)$ is constant on $0<p_{12}<1/2$.
    \item $\rho(P^{t}\mathcal{M})$ is decreasing on $(0,p_{21})$, increasing on $(p_{21},1/2)$ and attains one absolute minimum over $(0,1/2)$ with value $(\beta \beta_v)/(\delta \delta_v)$ at $p_{12}=p_{21}$.
    \item In addition, we have the inequality 
    
\begin{align}
\rho(P^{\T}\mathcal{M})-\rho(\mathcal{M}) & \leq   \frac{\beta \beta_v}{\delta \delta_v (2m_v + \delta_v)}|p_{12}-p_{21}|(1-p_{12}-p_{21})\delta_v\notag \\
& \leq   \frac{\beta \beta_v}{\delta \delta_v (2m_v + \delta_v)}\max\left\{p_{21}(1-p_{21})\delta_v,(1/2-p_{21})^2\delta_v\right \}\notag \\
&\leq \frac{\beta \beta_v}{\delta \delta_v (2m_v + \delta_v)}\frac{\delta_v}{4} \, .
\end{align}

\end{enumerate}
\end{proposition}

\begin{proof}
We have that $$ \mathcal{M}=D_{\beta}G_v^{-1}D_{\beta_v}P \left(D_\delta ^{\lag}\right)^{-1} = \frac{\beta \beta_v}{\delta \delta_v (2m_v + \delta_v)} \begin{pmatrix} m_v + (1-p_{21})\delta_v & m_v + p_{12}\delta_v \\ m_v + p_{21}\delta_v & m_v + (1-p_{12})\delta_v\end{pmatrix} \, ,$$

\noindent so the eigenvalues of $\mathcal{M}$ in this case are $$ \frac{\beta \beta_v}{\delta \delta_v} \quad\textrm{and}\quad \frac{\beta \beta_v}{\delta \delta_v} \frac{(1-p_{12}-p_{21})\delta_v}{2m_v+\delta_v}.$$

In consequence, we get $\rho(\mathcal{M}) = (\beta \beta_v)/(\delta \delta_v)$.
We can also get $\rho(P^{\T}\mathcal{M})$ explicitly by 
\begin{align}
\rho(P^{\T}\mathcal{M}) = \frac{\beta \beta_v}{\delta \delta_v (2m_v + \delta_v)}\left[m_v + \delta_v - \left(p_{12}(1-p_{12})+p_{21}(1-p_{21})\right)\delta_v + \right. \notag \\
\left.\sqrt{\left[m_v + \left(p_{12}(1-p_{12})+p_{21}(1-p_{21})\right)\delta_v\right]^2 + \left[(p_{12}-p_{21})(1-p_{12}-p_{21})\delta_v\right]^2} \right]\, .
\end{align}

From this equation, it follows that when $p_{12}=p_{21}$, we get
$$ \rho(P^{\T}\mathcal{M})= \frac{\beta \beta_v}{\delta \delta_v} = \rho(\mathcal{M})\, . $$

Moreover, $ \partial \rho(P^{\T}\mathcal{M})/\partial p_{12} = \kappa \beta \beta_v/\left[\delta \delta_v (2m_v + \delta_v)\right] $, where
\begin{align}
 \kappa = 2p_{12}-1 + 
\frac{m_v(1-2p_{21}) + \delta_v \left[2p_{12}(1-p_{12})(1-2p_{12}) + (p_{12}-p_{21})\left(2p_{12}(1-p_{12}-p_{21})+p_{12}+p_{21}\right)\right]}{\sqrt{\left[m_v + \left(p_{12}(1-p_{12})+p_{21}(1-p_{21})\right)\delta_v\right]^2 + \left[(p_{12}-p_{21})(1-p_{12}-p_{21})\delta_v\right]^2}}\,.
\end{align}

In particular, if $p_{12}=p_{21}$, then $$\frac{\partial \rho(P^{\T}\mathcal{M})}{\partial p_{12}}= 0 \, .$$

Assume that $p_{12} > p_{21}$. Define

\begin{align*}
\alpha & := \sqrt{\left[m_v + \left(p_{12}(1-p_{12})+p_{21}(1-p_{21})\right)\delta_v\right]^2 + \left[(p_{12}-p_{21})(1-p_{12}-p_{21})\delta_v\right]^2} \\ 
& \leq \left[m_v + \left(p_{12}(1-p_{12})+p_{21}(1-p_{21})\right)\delta_v\right] + \left[(p_{12}-p_{21})(1-p_{12}-p_{21})\delta_v\right] \, , 
\end{align*}

\noindent and  $$\eta:= m_v(1-2p_{21}) + \delta_v \left[2p_{12}(1-p_{12})(1-2p_{12}) + (p_{12}-p_{21})\left(2p_{12}(1-p_{12}-p_{21})+p_{12}+p_{21}\right)\right]\,.$$ We then have that
$$ \frac{\partial \rho(P^{\T}\mathcal{M})}{\partial p_{12}}= \frac{\beta \beta_v}{\delta \delta_v (2m_v + \delta_v)} \left[2p_{12}-1 + \frac{\eta}{\alpha}\right]$$ and
\begin{align*}
(2p_{12}-1)\alpha +\eta & \geq (2p_{12}-1)\left[m_v + \left(p_{12}(1-p_{12})+p_{21}(1-p_{21})\right)\delta_v + (p_{12}-p_{21})(1-p_{12}-p_{21})\delta_v\right] + \eta \\
& = 2(p_{12}-p_{21})\left[m_v +  \left( (p_{12}-1)^2 + p_{12}^2\right) \delta_v \right] > 0 \, . 
\end{align*}

Therefore, $p_{12} > p_{21}$ implies that $$\frac{\partial \rho(P^{\T}\mathcal{M})}{\partial p_{12}}= \frac{\beta \beta_v}{\delta \delta_v (2m_v + \delta_v)} \frac{(2p_{12}-1)\alpha + \eta}{\alpha} >0 \,. $$

Now, let us assume that $p_{12} <p_{21}$. We then have that
$$ \alpha \geq m_v + \left(p_{12}(1-p_{12})+p_{21}(1-p_{21})\right)\delta_v $$ and 
\begin{align*}
(2p_{12}-1)\alpha +\eta & \leq (2p_{12}-1)\left[m_v + \left(p_{12}(1-p_{12})+p_{21}(1-p_{21})\right)\delta_v\right]+ \eta \\
&=-(p_{21}-p_{12})(m_v + \delta_v ( 1 + (p_{21}-p_{12})(1-2p_{12}) ) )< 0\, .   
\end{align*}

Therefore, $p_{12} < p_{21}$ implies that $$\frac{\partial \rho(P^{\T}\mathcal{M})}{\partial p_{12}}= \frac{\beta \beta_v}{\delta \delta_v (2m_v + \delta_v)} \frac{(2p_{12}-1)\alpha + \eta}{\alpha} <0 \, . $$

Using that $\alpha \leq \left[m_v + \left(p_{12}(1-p_{12})+p_{21}(1-p_{21})\right)\delta_v\right] + |p_{12}-p_{21}|(1-p_{12}-p_{21})\delta_v $, we get 
\begin{align*}
\rho(P^{\T}\mathcal{M})-\rho(\mathcal{M}) & \leq   \frac{\beta \beta_v}{\delta \delta_v (2m_v + \delta_v)}|p_{12}-p_{21}|(1-p_{12}-p_{21})\delta_v\\
& \leq   \frac{\beta \beta_v}{\delta \delta_v (2m_v + \delta_v)}\max\left\{p_{21}(1-p_{21})\delta_v,(1/2-p_{21})^2\delta_v\right \}\\
&\leq \frac{\beta \beta_v}{\delta \delta_v (2m_v + \delta_v)}\frac{\delta_v}{4} \, .
\end{align*}

\end{proof}

Let us try to get some intuition for the inequality $\R_0^{lag} \geq \R_0^{eul}$ in the previous proposition. Suppose that $p_{21} < p_{12}$ and all the other parameters are assumed to be as in Proposition \ref{prep:homogeneous}. Assume that the systems (\ref{emodel1}) and (\ref{emodel2}) are at the DFE and suppose we introduce the same number of infectious hosts in both patches, say $I_h = I_{h,1} = I_{h,2}$. 
From the last equation in (\ref{emodel1}) of the Lagrangian system and A\ref{a:DFEs}, the rates at which vectors get infected in patches $1$ and $2$ are 

$$ \beta_{v,1} \frac{p_{11} I_{h,1} + p_{12} I_{h,2}}{ p_{11}\left(N_1^{\lag}\right)^* + p_{12}\left(N_2^{\lag}\right)^*} N_{v,1}^* = \beta_{v} I_h (1+p_{12}-p_{21})$$
and
$$ \beta_{v,2} \frac{p_{21} I_{h,1} + p_{22} I_{h,2}}{ p_{21}\left(N_1^{\lag}\right)^* + p_{22}\left(N_2^{\lag}\right)^*} N_{v,2}^* = \beta_{v} I_h(1+p_{21}-p_{12})$$
respectively. Since $p_{12}>p_{21}$, the vector infection rate in patch $1$ is greater than in patch $2$ (because $1+p_{12}-p_{21} > 1 > 1+p_{21}-p_{12}$). Therefore, over a period $\Delta t$, the number of infected vectors at patch $1$, which is approximately $\Delta I_{v,1}^{\lag}:= \beta_{v} I_h (1+p_{12}-p_{21})\Delta t$, would be larger than the amount of infected vectors at patch $2$, which is approximately $\Delta I_{v,2}^{\lag}:= \beta_{v} I_h (1+p_{21}-p_{12}) \Delta t$. From the second equation in (\ref{emodel1}), the amount of new host infections in patches $1$ and $2$ caused by the new infected vectors in DFE would be $$ \Delta I_{h,1}^{\lag} := \beta_1 p_{11}\left(\Delta I_{v,1}^{\lag}\right) + \beta_2 p_{21}\left(\Delta I_{v,2}^{\lag}\right) = \beta\left( p_{11}\left(\Delta I_{v,1}^{\lag}\right) +  p_{21}\left(\Delta I_{v,2}^{\lag}\right)\right) $$ and $$ \Delta I_{h,2}^{\lag} := \beta_1 p_{12}\left(\Delta I_{v,1}^{\lag}\right) + \beta_2 p_{22}\left(\Delta I_{v,2}^{\lag}\right)=\beta\left( p_{12}\left(\Delta I_{v,1}^{\lag}\right) +  p_{22}\left(\Delta I_{v,2}^{\lag}\right)\right) $$
respectively. Therefore, the total amount of new infected hosts would be
\begin{equation}\label{eqn:heuristics1}
   \Delta I_h^{\lag}:= \Delta I_{h,1}^{\lag}+\Delta I_{h,2}^{\lag} = \beta \left(\Delta I_{v,1}^{\lag} + \Delta I_{v,2}^{\lag} \right) +\beta  (p_{12}-p_{21})\left(\Delta I_{v,1}^{\lag}-\Delta I_{v,2}^{\lag}\right)\,. 
\end{equation}

On the other hand, from the last equation in (\ref{emodel2}) of the Eulerian system, the rates at which vectors get infected in patches $1$ and $2$ are $$ \beta_{v,1} \frac{I_{h,1}}{\left(N_1^{\eul}\right)^*} N_{v,1}^* = \beta_{v} I_h \quad\text{and}\quad \beta_{v,2} \frac{I_{h,2}}{\left(N_2^{\eul}\right)^*} N_{v,2}^* = \beta_{v} I_h$$
respectively. Therefore, over a period $\Delta t$, the amounts of new infected vectors are $\Delta I_{v,1}^{\eul} = \Delta I_{v,2}^{\eul} = \beta_v I_h \Delta t$. Notice that $$ \Delta I_{v,1}^{\lag} + \Delta I_{v,2}^{\lag} = 2 \beta_{v} I_h  \Delta t =  \Delta I_{v,1}^{\eul} + \Delta I_{v,2}^{\eul}\,. $$
From the second equation in (\ref{emodel2}) of the Eulerian system, the total amount of new infected hosts is

\begin{equation}\label{eqn:heuristics2}
    \Delta I_{h}^{\eul} =  \beta (\Delta I_{v,1}^{\eul} + \Delta I_{v,2}^{\eul}) = \beta (\Delta I_{v,1}^{\lag} + \Delta I_{v,2}^{\lag}) \, .
\end{equation}

 Since $p_{12}>p_{21}$ and $\Delta I_{v,1}^{\lag} > \Delta I_{v,2}^{\lag}$, we get that the term $\beta  (p_{12}-p_{21})\left(\Delta I_{v,1}^{\lag}-\Delta I_{v,2}^{\lag}\right)$ in (\ref{eqn:heuristics1}) is positive, and then using (\ref{eqn:heuristics2}), we have that $\Delta I_h^{\lag} > \Delta I_h^{\eul}$, which corresponds to $\R_0^{\lag} \geq \R_0^{\eul}$.  Note that this implicitly relies upon the fact that the removal rates, the rates $\beta_{v,i}$ and the rates $\beta_i$ are the same across patches.  As we observed in Figure \ref{figdeltas}, imposing different removal rates, for instance, can lead to $\R_0^{\eul} > \R_0^{\lag}$.

\bibliography{biblioLagEul}
\bibliographystyle{plain}
%\nocite{*}

\end{document}